\documentclass[12pt]{article}
  \usepackage[ ]{hyperref}
\usepackage{graphicx, amsmath,amsthm,amsopn,amsfonts,amssymb,dsfont, esint, oldgerm,yhmath,enumerate,bm}
\usepackage{color,yfonts}

\usepackage{authblk}
 \usepackage{pdfpages}
 \usepackage{graphicx}
  \usepackage{fullpage}
\usepackage{color,yfonts}
\usepackage{array}

\newcommand\mv[1]{\langle #1 \rangle}

\newcommand\ep{\varepsilon}

\newcommand\pd[2]{\tfrac{\partial #1}{\partial #2}}

\renewenvironment{thebibliography}[1]{\begin{list}{\arabic{enumi}.}
{\usecounter{enumi}\setlength{\parsep}{0pt}}}{\end{list}}

\newtheorem{thm}{Theorem}[section]
\newtheorem{prop}{Proposition}[section]
\newtheorem{lem}{Lemma}[section]

\newtheorem{rem}{Remark}[section]
\newtheorem{cor}{Corollary}[section]

\numberwithin{equation}{section}

\renewcommand{\i}{\mathrm{i}}
\renewcommand{\t}{\theta}

\newcommand{\RR}{\mathbb{R}}
\newcommand{\ZZ}{\mathbb{Z}}
\newcommand{\NN}{\mathbb{N}}
\newcommand{\CC}{\mathbb{C}}

\newcommand{\B}{\mathcal{B}}

\newcommand{\G}{\mathcal{G}}

\renewcommand{\H}{\mathcal{H}}
\newcommand{\norm}[1]{\| #1\|}

\newcommand{\e}{\varepsilon}
\renewcommand{\div}{\mathrm{div}\,}
\newcommand{\curl}{\mathrm{curl}\,}

\author[1,2]{\vspace{-.2cm}Shane Cooper}
\author[1]{Ilia V. Kamotski}
\affil[1]{\footnotesize Department of Mathematics, University College London, Gordon Street, London WC1E 6BT, UK.}
\affil[2]{\footnotesize Corresponding author. Email: s.cooper@ucl.ac.uk}

\title{{On photonic band gaps in two-dimensional photonic crystal fibres. Analysis in the vicinity of the low-dielectric light line.}}

\begin{document}
\maketitle
\begin{abstract} 
We consider  `off-axis' electromagnetic wave propagation down the homogeneous direction of a low-loss two-dimensional  periodic dielectric (or photonic crystal fibre) near the light line of the low-dielectric material constituent. Numerous physical and numerical experiments demonstrate the presence of photonic band gaps in the vicinity of this `critical' light line.  We mathematically analyse the existence of  photonic band gaps near the line and characterise them in terms of frequency gaps in the spectrum of the Maxwell equations restricted to the line. We apply the results to both one-dimensional photonic crystal fibres, and a genuinely two-dimensional photonic crystal fibres with `thin' inclusions, namely ``ARROW'' fibres. In the case of ARROW fibres, by an asymptotic analysis,  in terms of the small inclusion parameter, we demonstrate the existence of low frequency photonic band  gaps. It is important to note that our analysis does not assume any specific ratio between the dielectric contrasts; as a result, moderate or even low contrast models are fully encompassed. Similarly, no limiting assumptions are imposed on the wave propagation constant along the fibre.
\end{abstract}

\section{Introduction}
A two-dimensional photonic crystal fibre is a low-loss periodic dielectric with material properties that are homogeneous along one of its axes. Photonic crystal fibres (PCFs) are known to possess photonic band gaps; that is, some PCFs prevent certain frequency ranges of light from propagating in some, or all, directions. Such devices are of practical importance, as they can be used to construct waveguides with optical performances that surpass traditional optical waveguides (which operate on the principle of total internal reflection), such as omnidirectional reflectors and narrow-band filters (see \cite{MFL} and references therein). Consequently, an important mathematical task is to analytically characterise photonic band gaps in terms of a PCF's material and geometric properties.

In physical literature, the origin of a photonic band gap is explained as follows (cf. \cite{Ru03}): suppose we consider plane-wave solutions to the Maxwell system with frequency $\omega$ and wave number $k$ down a given axis, through a given homogeneous dielectric material with constant dielectric permittivity $\epsilon$. Then, the $(\omega,k)$-plane is divided into two regions by the `light line' $\omega(k) = c k / \sqrt{\epsilon}$, for $c$ the speed of light; above this line, waves propagate in the direction perpendicular to the axis, and below this line, waves are evanescent in the perpendicular direction. For propagation down a photonic crystal fibre, with wave number $k$ and frequency $\omega$, the key observation of the above argument is this: if $(\omega,k)$ is above the light line for the low-dielectric matrix phase, the whole PCF can support the wave, but just below the low-dielectric matrix light line (and above the high-dielectric inclusion light line) the matrix phase can no longer support propagating EM waves. It is argued that, for this reason, one should expect photonic band gaps to appear near the low-dielectric inclusion light line. This heuristic argument is supported in literature through numerous physical and numerical experiments (\cite{Ru03}, \cite{2Lim}, \cite{FoPCF}). We henceforth refer to the low-dielectric inclusion light line as the critical (light) line. 

In this work we provide a rigorous quantitative mathematical analysis of the above heuristics. Before discussing this article's results, we review previous related mathematical results.  The rigorous mathematical theory of band gaps has largely focused on trivial wave number $k$ regimes or the high-contrast between material parameters regime. In classical theory, $k=0$ is assumed and the Maxwell system decouples to scalar Helmholtz equation due to high symmetries that allow polarisation. In the setting $k \rightarrow \infty$ the Maxwell system asymptotically decouples once more to the scalar regime and this was rigorously analysed in \cite{Bon} via asymptotic methods. As for the high-contrast regime, various asymptotic approaches have been used to yield explicit spectral gap results (e.g., \cite{FiKu}; \cite{AmKaLi}, \cite{BoBoFe},  \cite{zhikov2}) and the \cite{LiVi} considered how large the contrast must be to induce spectral gaps.

In this work, we do not assume the high-contrast regime; in fact, for non-trivial wave propagation down the photonic crystal fibre, i.e., $k\neq 0$ nor $k\not\rightarrow \infty$ , the EM field does not admit polarisations; consequently, EM propagation directly above the critical line reduces the Maxwell system to a system of two second-order elliptic differential spectral equations with coefficients that artificially" become highly contrasting and an ellipticity constant that degenerates to zero as one approaches the critical light line ($k \rightarrow 1$). In \cite{SCth}, the author developed further two-scale resolvent convergence methods, of \cite{zhikov1}, \cite{zhikov2}, and \cite{IVKVPS}, to analyse the asymptotics of the Bloch spectrum of the associated degenerating high-contrast differential operator, as one approached the critical light line. It was shown that the Bloch spectrum converged to the spectrum of some `limit' operator, and that photonic band gaps near the critical line will appear if the determined limit operator's spectrum possessed gaps.

Here, we improve on the results of \cite{IVKVPS} in several key ways. First, we establish that the limit operator, found in \cite{SCth}, is precisely the Maxwell operator restricted to the critical line. Furthermore, we utilise some methods of \cite{SCIKVPS23} to improve the above convergence results to provide operator-type error estimates between the Maxwell operators on and above the critical light line, respectively, in terms of the distance to the line. This implies, among other things, that if the Maxwell operator on the critical light line has gaps in its spectrum, then photonic band gaps exist in a quantifiable neighbourhood of the restricted Maxwell spectral gaps. Finally, we study the particular problem of ARROW fibres, i.e., PCFs with thin inclusions, to demonstrate the existence of low-frequency gaps in this restricted Maxwell spectrum and therefore, by the above results, prove the existence of low-frequency photonic band gaps for ARROW PCF.

 More precisely, the results and structure of the article are as follows. We consider, see Section \ref{secPr}, electromagnetic waves
\begin{equation*}
	e^{{\rm i}\omega (kx_3 - t)}E(x_1,x_2), \ \ \ \  e^{{\rm i}\omega (kx_3-t)}H(x_1,x_2), 
\end{equation*} 
 with $\theta$-quasiperiodic  electric and magnetic amplitudes, $E$ and $H$ respectively, 
propagating down the homogeneous $x_3$-direction of non-magnetic two-dimensional PCF with (normalised) dielectric permittivity $\epsilon$ equal to $\epsilon_0>1$ in the inclusion phase and $1$ in the surrounding `matrix' phase.

A photonic band gap, by definition, consists of non-discrete regions (or `holes') in the  $(\omega, k)$-plane for which there are no non-trivial solutions of the above form  to the Maxwell system.  Now, such EM waves are solutions to the Maxwell system (for speed of light $c=1$) if, and only if, the amplitudes $E$ and $H$ satisfy the following system:
 
 \begin{equation}
 	\label{iMax1}
 	\begin{aligned}
 		\i 	\omega	\left( \begin{array}{cc}
 			k & - \epsilon  \\-	1  & k
 		\end{array}\right) \left( \begin{array}{c}
 			H_1 \\  E_2
 		\end{array}\right) = \left( \begin{array}{c}
 			H_{3,1} \\E_{3,2} 
 		\end{array}\right), & \hspace{.5cm} &
 		\i	\omega	\left( \begin{array}{cc}
 			k & \epsilon\\ 1 & k 
 		\end{array}\right) \left( \begin{array}{c}
 			H_2 \\ E_1
 		\end{array}\right) = \left( \begin{array}{c}
 			H_{3,2} \\E_{3,1}  
 		\end{array}\right),
 	\end{aligned}
 \end{equation}
 and
 \begin{equation*}
 	\begin{aligned}
 		E_{1,2} - E_{2,1}&=& \i \omega  H_3,& \hspace{1cm} &  	 H_{2,1} - H_{1,2} = \i \omega {\epsilon} E_3.
 	\end{aligned}
 \end{equation*}
   In our notation, being on the critical (low-dielectric constituent) light line corresponds to $k=1$ and being above this line corresponds to $k<1$. For $k=1$, the matrices in \eqref{iMax1} are clearly not invertible in matrix constituent (where $\epsilon =1$). This Maxwell system is `degenerate' and, in  Section \ref{s:k1}, we establish that plane wave solutions solve the Maxwell system on the critical light line ($k=1$)  if, and only if, $u = (H_3,E_3)$ solves the spectral problem
   \[
   \mathcal{B}(\t) u = \omega^2 u,
   \]
for some positive differential operator $\mathcal{B}(\t)$ whose domain is a subset of  $\t$-quasiperiodic vector fields $\phi$ that satisfy the Cauchy-Riemann-type constraints
\begin{equation}\label{iCR}
    \phi_{1,1} +  \phi_{2,2} = 0 \quad \text{and} \quad  \phi_{2,1}- \phi_{1,2}= 0 \quad \text{in $\square \backslash R$.}
\end{equation}
   Here $\square$ and $R$ are, respectively, the (two-dimensional) periodic reference cell and inclusion cross section of the PCF.    The above constraint is due to the non-trivial kernels of the matrices in \eqref{iMax1}; the fact that $(H_3,E_3)$ satisfies this constraint is precisely the Fredholm alternative for the algebraic systems in \eqref{iMax1}. 
   
   Above the critical line, the matrices in \eqref{iMax1} are invertible and, just above the critical line, precisely for $k=\sqrt{1- \epsilon}$,  plane wave solutions solve the Maxwell system if, and only if, $u=(H_3,E_3)$  solves an (unconstrained) spectral problem
    \[
    \mathcal{B}_\ep(\t) u = \omega^2 u,
    \]
    for some positive elliptic second-order differential operator (see Section \ref{s2.2}).      We prove, in Section \ref{s.resolventestimates}, that the Maxwell system above the critical line converges to the Maxwell system on the line in a nice continuous manner; or  more precisely, the pseudo-resolvent of $\B(\t)$ and resolvent of  $\B_\ep(\t)$ are, unformily in $\t$, close with respect to $\ep$ in the uniform operator topology:
    \begin{lem}
There exists a constant $C>0$ such that 
\[
\| (\B_\ep(\t)+I)^{-1} -(\B(\theta)+I)^{-1} P  \|_{L^2(\square;\CC^2) \rightarrow L^2(\square;\CC^2)} \le C \ep, \quad \forall \t \in \square^*, \forall\, 0<\ep<1.
\]
    \end{lem}
    Here, $P$ is a projection arising from the degeneracy of the Maxwell system on the critical line  (i.e. the above Cauchy-Riemann-type constraint).  This operator estimate, amongst other things, implies that the Bloch spectrum $\sigma_\ep$, of the Maxwell system on lines  above the critical line, converges with rate to the Bloch  spectrum $\sigma_0$, of the Maxwell system on the critical line, as you approach it: 
 \begin{lem}
 	For every compact $[a,b]$ in $\RR$, 
 	one has
 	\[
 	d_{[a,b]} \left( \sigma_\ep, \sigma_0\right) \le C_b \ep,
 	\]
 	for some positive constant $C_b$ independent of $\ep$ and $a$.  Here,
 	\[
 	d_{[a,b]}(X,Y) : = \max\{{\rm{dist}([a,b] \cap X, Y)}, \rm{dist}( [a,b]\cap Y,X) \}
 	\]
 	where $\rm{dist}(X,Y) := \sup_{x\in  X} \rm{dist}(x,Y)$ is the non-symmetric distance of  set $X$ from set $Y$ (adopting the usual convention that $\rm{dist}(\emptyset,A)=\rm{dist}(A,\emptyset) =0$ for any set $A$).
 \end{lem}

 An immediate consequence of this spectral estimate is our first main result, which states that, if there are frequencies of light that are forbidden to propagate on the critical light line (i.e., there are gaps in the spectrum $\sigma_0$) then a quantified neighbour of these frequencies are  photonic band gaps:
 \begin{thm}\label{tocite}
 	Assume that $I=(a,b)$ is a gap of $\sigma_0$. Then, the set
 \[
 \G = \bigcup_{0\le \ep < \min\{1,(b-a)/2C_b\}} \{ (\omega, \sqrt{1-\ep}) \, | \, \omega^2 \in (a+C_b \ep, b - C_b\ep) \}
 \]
 is a photonic band gap.  
 \end{thm}
 
 This result naturally leads to the question if whether or not $\sigma_0$ has gaps for a given two-dimensional PCF; this is the topic of the remainder of the  article.
 
  As $\mathcal{B}(\theta)$ has discrete spectrum, denoted $\{ \lambda_n(\t) \}_{n=1}^\infty$, then determining gaps boils down to studying neighbouring spectral band functions $E_n(\t) : = \lambda_n(\t)$. Unlike the standard situation of Floquet-Bloch operators associated with periodic differential operators defined in the whole space (where the spectral band functions are easily seen to be  continuous) it is not clear if $E_n$ are continuous in $\t$; this is due to the fact that  functions in the domain of $\mathcal{B}(\t)$ are $\t$-quasiperiodic functions that satisfy the constraint \eqref{iCR}. In particular, constant functions are in the domain of $\mathcal{B}(0)$ and it is a non-trivial fact that constants can be approximated by elements of $\mathrm{dom}(\mathcal{B}(\t))$ for small $\t$.  We prove, in Section \ref{secblochlim}, that the family of spaces $\mathrm{dom}(\mathcal{B}(\t)$ (and consequently the operators $\mathcal{B}(\t)$, and the spectral band functions) are continuous in  $\t$. 
 Turning back to the question of gaps, the regularity of the band functions gives the representation
 \[
 \sigma_0  = \overline{\sum_{n=1}^\infty [a_n,b_n]}, \quad a_n : = \min_{\t\in\square^*} \lambda_n(\t), b_n : = \max_{\t\in\square^*} \lambda_n(\t),
 \]
 Thus, for a general two-dimensional PCF, determining gaps in $\sigma_0$ (and consequently photonic band gaps near the critical line) is equivalent to establishing  $b_n < a_{n+1}$ for some $n$. In general,   it is unclear if gaps exist.
 
 In Section   \ref{s.1d}, we demonstrate, for one-dimensional photonic crystal fibres, that $\mathcal{B}(\t)$ is the one-dimesional Laplacian with mixed boundary conditions (that contain  the spectral parameter), and we demonstrate that gaps in $\sigma_0$ appear at the edges of the Brillouin zone for this non-standard Floquet problem.   For  genuinely two-dimensional PCFs, it is unclear if gaps exists in the spectrun $\sigma_0$; however, for the case of a two-dimensional PCF with small inclusions; namely $R = \delta \Omega$, for some domain $\Omega \subset \RR^2$ and small parameter $\delta$, we establish   the following asymptotics in Section \ref{s.Arrow}:
 \begin{thm}One has the following, uniform in $\t$, asymptotic behaviour:
 	\[
 	\lambda_{n}(\t) =\frac{1}{\delta^{2}\big( g_\t(0) - (2\pi)^{-1} \ln{\delta}\big)} \Lambda_n + \mathcal{O}(\delta^{-2}|\ln{\delta}|^{-3/2}), \quad n=1,2.
 	\]
 	and
 	\[
 	\lambda_n(\t) =\delta^{-2} \Lambda_n + \mathcal{O}(\delta^{-2} |\ln{\delta}|^{-1/2}), \quad n \ge 3, 
 	\]
 	as $\delta \rightarrow 0$.
 \end{thm}
Here, $g_\t(0)$ is the regular part of the $\t$-quasi periodic Green's function, evaluated at the origin; the numbers  $\Lambda_n$ (given precisely in Section \ref{examples}) are the eigenvalues of some particular $\theta$ and $\delta$ independent operator. Given that  the first two bands grow asymptotically slower than the higher bands,  this asymptotic analysis  immediately implies  that there is a low frequency gap in $\sigma_0$ (and therefore a low frequency photonic band gap) for sufficiently small $\delta$:
 \begin{cor}
 	There exists a $\delta_0>0$ such that
$
 b_2 < a_3,
$
 	for all $\delta<\delta_0$. Therefore $(a,b)$ is a gap in the spectrum $\sigma_0$ and, by Theorem \ref{tocite},  $\G$ is a photonic band. 
 \end{cor}
  We note here, that for the case of cylindrical inclusions, i.e. $\Omega$ the ball of radius one centered at the origin, one has 
 \[
 \Lambda_1  = \Lambda_2 =  \big(\frac{1}{\epsilon_0-1}+\frac{1}{2}\big),
 \]
 that is  the first two spectral band functions $\lambda_1(\t)$, $\lambda_2(\t)$ asymptotically coincide for small $\delta$.

The Appendix contains the proofs of the two key theorems highlighted above and other intermediary technical lemmas.
\section{Problem formulation}
\label{secPr}
We consider a two-dimensional non-magnetic periodic dielectric photonic crystal fibre that is homogeneous in the $x_3$ co-ordinate direction and has optically denser inclusions, see Figure \ref{fig:pcf}.
\begin{figure}
	\centering
	\includegraphics[width=0.7\linewidth]{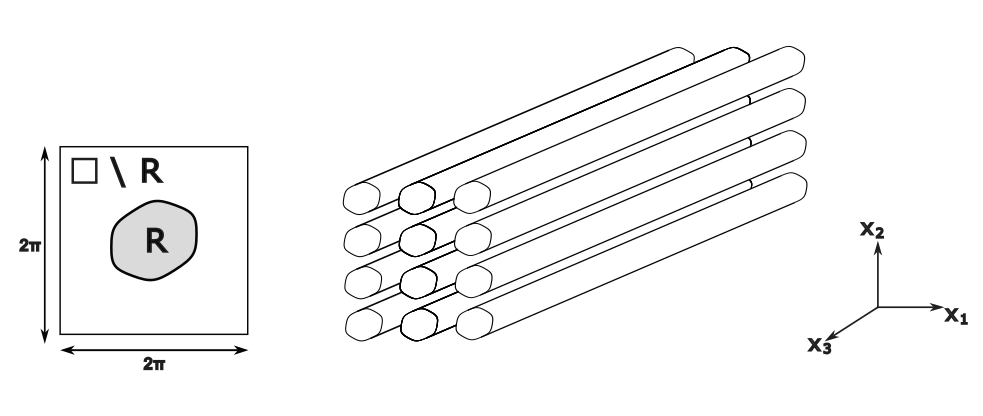}
	\label{fig:pcf}
	\caption{A two-dimensional photonic crystal fibre that is homogeneous in the $x_3$-direction with a cylindrical cross-section $R$ that is  repeated periodically with respect to the reference cell $\square=[-\pi,\pi)$ in the $(x_1,x_2)$-plane.}
\end{figure}

 That is, the photonic crystal fibre has magnetic permeability\footnote{Considering non-magnetic material is made for simplicity of the exposition. The subsequent  analysis and results hold for real  $\mu \neq 1$.} $\mu \equiv 1$  and a real-valued scalar dielectric permittivity  $\epsilon$ that is constant in $x_3$ and  periodic in the cross-sectional $(x_1,x_2)$-plane with respect to the periodic reference cell $\square : = [-\pi,\pi)^2$, i.e.
\begin{equation}\label{coef}
	\begin{aligned}
		\epsilon\big(x_1,x_2,x_3\big) =\left\{ \begin{matrix}
			\epsilon_0 & & \text{if $(x_1,x_2)\in \bigcup_{z\in \ZZ^2} (R+z) $,} \\[0.5em]
			1 & & \text{otherwise,}
		\end{matrix} \right.
		\end{aligned}
	\end{equation}
for\footnote{Without loss of generality, we have set the dielectric permittivity $\epsilon_1$ of the background media is unity; for $\epsilon_1 \neq 1$ we can simply consider the rescaled electric field $\sqrt{\epsilon_1} E$ and re-scaled frequency $\sqrt{\epsilon_1} \omega$ to arrive at the formulation \eqref{coef}-\eqref{propsol}.} $\epsilon_0>1$ and $R$ open such that $\overline{R} \subset \square$. We are interested in establishing the existence of electro-magnetic waves that  propagate down the fibres.  That is, for given  $\omega \in \RR$, $k \in [0,\infty)$ and $\t \in \square^* : = [-\frac{1}{2},\frac{1}{2})^2$,
we are interested in establishing the existence of solutions of the form\footnote{In our notation, $\omega$ is the frequency and $\omega k$ is the wave number, i.e. $k$ is the reciprocal of wave speed. This choice of notation is natural as we shall be considering straight lines in the frequency-wave number parameter space.} 
\begin{equation} \label{propsol}
	\tilde{E}(x_1,x_2,x_3)=e^{{\rm i}\omega (kx_3 - t)}E(x_1,x_2), \ \ \ \  \tilde{H}(x_1,x_2,x_3)=e^{{\rm i}\omega (kx_3-t)}H(x_1,x_2), 
\end{equation} 
with  $\t$-quasi-periodic amplitudes $E,H$,  to  the Maxwell system 
\begin{equation}
	\label{Maxwellsystem}
	\begin{aligned}
		\nabla \times \tilde{H} &= - \, \epsilon\, \pd{\tilde{E}}{t}, \\
		\nabla \times \tilde{E} &=  \, \pd{\tilde{H}}{t}.
	\end{aligned}
\end{equation}
Recall that $E,H$ being $\t$-quasi-periodic means
\begin{equation}\label{qpbc}
E(y + e_j) = e^{\i \t_j} E(y),  \quad \& \quad H(y+e_j) = e^{\i \t_j} H(y), \quad y \in \square, \, j=1,2,
\end{equation}
where $e_j$ is $j^\mathrm{th}$ Euclidean basis vector, or  equivalently, $E,H$ are of the form $e^{\i \t \cdot y} U(y)$ for some $\square$-periodic complex-valued vector field $U$.

Upon substituting \eqref{propsol} into \eqref{Maxwellsystem}, under \eqref{coef}, we find that $E=(E_1,E_2,E_3)$, $H=(H_1,H_2,H_3)$ necessarily solve  the following system of equations 

\begin{equation}
	\label{Max1}
	\begin{aligned}
\i 	\omega	\left( \begin{array}{cc}
		k & - \epsilon  \\-	1  & k
		\end{array}\right) \left( \begin{array}{c}
			H_1 \\  E_2
		\end{array}\right) = \left( \begin{array}{c}
			H_{3,1} \\E_{3,2} 
		\end{array}\right), & \hspace{.5cm} &
\i	\omega	\left( \begin{array}{cc}
			 k & \epsilon\\ 1 & k 
		\end{array}\right) \left( \begin{array}{c}
			 H_2 \\ E_1
		\end{array}\right) = \left( \begin{array}{c}
		H_{3,2} \\E_{3,1}  
		\end{array}\right),
	\end{aligned}
\end{equation}
and
\begin{equation}
	\label{Max2}
	\begin{aligned}
 E_{1,2} - E_{2,1}&=& \i \omega  H_3,& \hspace{1cm} &  	 H_{2,1} - H_{1,2} = \i \omega {\epsilon} E_3,
	\end{aligned}
\end{equation}
where $f_{,j}$ denotes the partial derivative of $f$ with respective to the $j^\mathrm{th}$ variable.

\subsection{The Maxwell system on the critical light line.}\label{s:k1}
In what follows, 
\[
\div u :=  u_{1,1} +  u_{2,2} \quad \text{and} \quad \curl u  :=u_{2,1}- u_{1,2}
\]
respectively denote  the two-dimensional divergence and curl of a vector-valued function $u$. We shall show that solutions of the form \eqref{qpbc} solve the Maxwell system, when $k=1$, if, and only if,  the $\t$-quasi-periodic vector field $u = (H_3,E_3)$ satisfying 
\begin{equation}\label{uCR}
\div u = 0 \quad \text{and} \quad \curl u = 0 \quad \text{in $\square \backslash R$,}
\end{equation}
and
\begin{equation}\label{degp}
b(u,\phi) = \lambda\, c(u,\phi) ,  \qquad \text{where\ } \lambda = \omega^2,
\end{equation}
for all $\t$-quasi-periodic test functions $\phi$ such that $\div{\phi}=\curl{\phi} =0 $ in $\square\backslash R$. Here, $b$, $c$ are the symmetric sesquilinear forms
\begin{equation}\label{prebform}
	\begin{aligned}
b(u,\phi) &: =  -\frac{1}{2} \int_{\square \backslash R} \Big( \nabla^\perp u_1 \cdot \overline{\nabla \phi_2} +\nabla u_2 \cdot\overline{\nabla^\perp \phi_1} \Big)  +\\&\frac{1}{\gamma} \int_R \Big( \nabla u_1 \cdot \overline{\nabla \phi_1} + \epsilon_0 \nabla u_2 \cdot \overline{\nabla \phi_2}  +   \nabla^\perp u_1 \cdot \overline{\nabla \phi_2} +\nabla u_2 \cdot\overline{\nabla^\perp \phi_1} \big),
	\end{aligned}
\end{equation}
and
\begin{equation}\label{cform}
c(u,\phi) :=   \int_{\square} u \cdot \overline{\phi} + \gamma \int_R  u_2  \overline{\phi_2},
\end{equation}
for positive constant  $\gamma : = \epsilon_0-1$ and for  $\nabla^\perp$ denoting the rotation of the gradient vector by $\pi/2$, i.e. $\nabla^\perp p = (-p_{,2}, p_{,1})$.

We shall show sufficiency here, and necessity is left until Lemma \ref{BthetaIsmaxwell} in the appendix.

Multiplying the first and second equations in \eqref{Max2} by $\t$-quasi-periodic test functions $\phi_1$  and $\phi_2$ respectively, integrating these equations over the periodic cell $\square$ and adding them together gives, after a simple application of integration by parts, the following identity:
\begin{equation*}
\int_\square  \left( \begin{array}{c}
	H_1 \\  E_2
\end{array}\right)  \cdot \overline{ \left( \begin{array}{c}
	\phi_{2,2} \\  \phi_{1,1}
\end{array}\right) } -  \left( \begin{array}{c}
H_2 \\  E_1
\end{array}\right)  \cdot \overline{ \left( \begin{array}{c}
	\phi_{2,1} \\  \phi_{1,2}
\end{array}\right) } = \i \omega\left( \int_\square  H_3 \overline{\phi_1}+E_3 \overline{\phi_2} + \gamma \int_R E_3 \overline{\phi_2} \right).
\end{equation*}
Since the systems in \eqref{Max1} degenerate in $\square \backslash R$ (as $k=1$), we study the left integral, above, in the regions $\square \backslash R$ and $R$ separately.

\textbullet\,  In $\square\backslash  R,$ we notice that the matrices in \eqref{Max1} have non-trivial kernels. This provides constraints on $E_3,H_3$.   Indeed, we see that
\[
	\left( \begin{array}{cc}
	k & - \epsilon \\ -1  & k
\end{array}\right) = 	\left( \begin{array}{cc}
1  & -1 \\-1 &  1
\end{array}\right )\quad \text{in $\square \backslash R$}
\]
which has zero as an eigenvalue with eigenvector $ \left( \begin{array}{c}
	1 \\  1
\end{array}\right)$. Therefore, 
$ \left( \begin{array}{c}
	H_{3,1}\\ E_{3,2} 
\end{array}\right)$ must be perpendicular to the vector $ \left( \begin{array}{c}
1 \\  1
\end{array}\right)$ in $\square\backslash R$, i.e.
\begin{equation*}\label{CR1}
H_{3,1}+E_{3,2}  = 0 \quad \text{ in $\square \backslash R$}.
\end{equation*}
Then, from \eqref{Max1}, it follows that 
\[
 \left( \begin{array}{c}
	H_1 \\  E_2
\end{array}\right) = -\frac{\i}{\omega} \frac{1}{2}  \left( \begin{array}{c}
H_{3,1} \\  E_{3,2}
\end{array}\right) + C_1  \left( \begin{array}{c}
1 \\  1
\end{array}\right) \quad \text{in $\square \backslash R$}
\]
for some function $C_1$. Similarly, by the second system in \eqref{Max1}, it follows that  
\begin{equation*}\label{CR2}
E_{3,1}  -H_{3,2}= 0 \quad \text{ in $\square \backslash R$},
\end{equation*}
and
\[
\left( \begin{array}{c}
	H_2 \\  E_1
\end{array}\right) = -\frac{\i}{\omega} \frac{1}{2}  \left( \begin{array}{c}
	H_{3,2} \\  E_{3,1}
\end{array}\right) + C_2  \left( \begin{array}{c}
-	1 \\  1
\end{array}\right) \quad \text{in $\square \backslash R$}
\]
for some $C_2$.
Consequently, for $\phi = (\phi_1,\phi_2)$ satisfying 
\[
\div \phi = 0, \quad \curl \phi = 0, \quad \text{in $\square \backslash R$,}
\] 
a simple calculation gives 
\begin{equation*}
\int_{\square \backslash R}  \left( \begin{array}{c}
	H_1 \\  E_2
\end{array}\right)  \cdot \overline{ \left( \begin{array}{c}
		\phi_{2,2} \\  \phi_{1,1}
	\end{array}\right) } -  \left( \begin{array}{c}
	H_2 \\  E_1
\end{array}\right)  \cdot \overline{ \left( \begin{array}{c}
		\phi_{2,1} \\  \phi_{1,2}
	\end{array}\right) } =  -\frac{\i}{\omega} \frac{1}{2} \int_{\square \backslash R} \Big( \nabla^\perp u_1 \cdot \overline{\nabla \phi_2} +\nabla u_2 \cdot\overline{\nabla^\perp \phi_1} \Big) 
\end{equation*}
for $u = (H_3,E_3)$.

\textbullet \, In $R$, the matrices in  \eqref{Max1} are invertible, so:
\[
 \left( \begin{array}{c}
	H_1 \\  E_2
\end{array}\right) = \frac{\i}{\omega}\frac{1}{\gamma}\left( \begin{array}{c}
H_{3,1} +\epsilon_0E_{3,2} \\ H_{3,1}+E_{3,2}
\end{array}\right), \quad   \left( \begin{array}{c}
H_2 \\  E_1
\end{array}\right) = \frac{\i}{\omega}\frac{1}{\gamma}\left( \begin{array}{c}
H_{3,2}-\epsilon_0E_{3,1} \\ E_{3,1}-H_{3,2}
\end{array}\right), \quad \text{in $R$}.
\]
Then, a simple computation (for any test function $\phi$) gives
\begin{equation*}
	\begin{aligned}
\int_{R}  \left( \begin{array}{c}
	H_1 \\  E_2
\end{array}\right)  \cdot \overline{ \left( \begin{array}{c}
		\phi_{2,2} \\  \phi_{1,1}
	\end{array}\right) } -  & \left( \begin{array}{c}
	H_2 \\  E_1
\end{array}\right)  \cdot \overline{ \left( \begin{array}{c}
		\phi_{2,1} \\  \phi_{1,2}
	\end{array}\right) }  \\
&= \frac{\i}{\omega}\frac{1}{\gamma} \int_R \Big( \nabla u_1 \cdot \overline{\nabla \phi_1} + \epsilon_0 \nabla u_2 \cdot \overline{\nabla \phi_2}  +   \nabla^\perp u_1 \cdot \overline{\nabla \phi_2} +\nabla u_2 \cdot\overline{\nabla^\perp \phi_1} \big).
	\end{aligned}
\end{equation*}

Putting all this together, shows that $u= (H_3,E_3)$ satisfies \eqref{uCR} and \eqref{degp}.

\subsection{The Maxwell system above the critical light line.}\label{s2.2}
Here, we consider $k = \sqrt{1 -\ep}$ for fixed $0<\ep<1$.  We shall show that the problem \eqref{qpbc}-\eqref{Max2} is equivalent  to $\t$-quasi-periodic $u = (H_3,E_3)$ and $\lambda = \omega^2$ satisfying 
\begin{equation}\label{nondp}
\ep^{-1} a(u,\phi) + b_{\ep}(u,\phi) = \lambda c(u,\phi)
\end{equation}
for the symmetric\footnote{Symmetry of $b_\ep$ can be seen, for example, from the identity $\nabla f\cdot \nabla^\perp g = - \nabla g \cdot \nabla^\perp f$.} sesquilinear forms
\begin{equation}\label{aform}
	\begin{aligned}
		a(u,\phi) : = \int_{\square \backslash R} \div u\, \overline{\div\phi} + \curl {u} \,\overline{\curl \phi},
	\end{aligned}
\end{equation}
\begin{equation}\label{bepform}
	\begin{aligned}
		&b_{\ep}(u,\phi) :=	\int_{\square \backslash R} \frac{\sqrt{1-\ep}-1}{\ep}\Big(  \nabla^\perp u_1 \cdot \overline{\nabla \phi_2} + \nabla u_2 \cdot \overline{\nabla^\perp \phi_1}\Big)\  +\\
		&\frac{1}{\gamma+ \ep} \int_{ R} \Big(  \nabla  u_1 \cdot \overline{\nabla \phi_1} + \epsilon_0 \nabla  u_2 \cdot \overline{\nabla \phi_2} 
		+ \sqrt{1-\ep}( \nabla^\perp u_1 \cdot \overline{\nabla \phi_2} + \nabla u_2 \cdot \overline{\nabla^\perp \phi_1})\Big),
	\end{aligned}
\end{equation}
and $c$ as in \eqref{cform}.

To show equivalence, note that the matrices in \eqref{Max1} are everywhere invertible and upon solving \eqref{Max1} for $E_1,E_2,H_1$ and $H_2$, and then substituting the resulting expressions into \eqref{Max2} leads to the following PDE  system for $(H_3,E_3)$:
\begin{align}
	 -\Big(\frac{1}{A} H_{3,1} \Big)_{,1} - \Big(\frac{1}{A} H_{3,2} \Big)_{,2} - \Big(\frac{k}{{A}} E_{3,2} \Big)_{,1} + \Big(\frac{k}{{A}} E_{3,1} \Big)_{,2}& =   \omega^2  H_3,  \label{14} \\
- \Big(\frac{ {\epsilon}}{{A}} E_{3,1} \Big)_{,1} - \Big(\frac{{\epsilon}}{{A}} E_{3,2} \Big)_{,2} +\Big(\frac{k}{{A}} H_{3,2} \Big)_{,1} - \Big(\frac{k}{{A}} H_{3,1} \Big)_{,2}  & =   \omega^2 {\epsilon} E_3  \label{13} 
\end{align} 
for 
\begin{equation*}
	A \,\, : = \left\{ \begin{matrix}
\ep & & \text{in}\ \square \backslash R, \\
\gamma + \ep  & & \text{in}\  R,
	\end{matrix} \right. 
\end{equation*} 
where we recall $\gamma = \epsilon_0 -1>0$. 
Upon setting $u=(H_3,E_3)$,  multiplying \eqref{14} and \eqref{13} by $\t$-quasi-periodic test functions $\phi_1$  and $\phi_2$, respectively, integrating over the periodic cell $\square$ and adding the resulting integral equations together, along with  the identity
\begin{equation}\label{15.6.24e2}
	\div{\phi} \,\div{\psi} +\curl \phi \, \curl{\psi} = \nabla \phi : \nabla \psi+  \nabla^\perp \phi_1 \cdot \nabla \psi_2 + \nabla \phi_2 \cdot \nabla^\perp \psi_1,
\end{equation} gives \eqref{nondp}.
\subsection{Operator theoretic formulation of the spectral problems.}\label{s:o}

\textbf{On the critical line: $k=1$.} Here, we rewrite  the spectral problems \eqref{degp} and \eqref{nondp} in operator-theoretic language. To that end, we recall the Lebesgue space $L^2(\square;\CC^2)$ of  complex vector-valued functions whose components belong to $L^2(\square)$; we recall the Sobolev spaces  $H^1(\square; \CC^2)$ of complex  vector-valued functions whose components belong to $H^1(\square)$, the closed subspace $H^1_{per}(\square;\CC^2)$  of  $H^1(\square;\CC^2)$ functions that are $\square$-periodic and  $H^1_\t(\square;\CC^2)$  the Sobolev space of  $\t$-quasi-periodic $H^1(\square;\CC^2)$ functions, i.e. $H^1_\t(\square;\CC^2) : = \{ e^{\i\t \cdot y}v(y) \, | \, v \in H^1_{per}(\square;\CC^2)\}$. All these spaces are equipped with the standard norms. 

We introduce  the family of linear subspaces
\begin{equation}
	\label{Vtheta}
	V(\t) : = \{ v\in H^1_\t(\square;\CC^2)\ | \ \div v=0 \quad \& \quad \curl  v =0  \ \text{in $\square\backslash R$}\}, \quad \t \in \square^*. 
\end{equation}
\begin{rem}
	$v \in V(\theta)$ is `equivalent' to the complex function $F(y_1 + \i y_2) = v_1(y) - i v_2(y)$  being complex-differentiable in $\square \backslash R$.
\end{rem}
Notice that, for each $\t$,  $V(\theta)$ is a closed subspace of $H^1_\theta(\square;\CC^2)$. Furthermore,  by \eqref{Vtheta}, one has
\[
 \nabla^\perp \tilde{v}_1 = - \nabla \tilde{v}_2  \ \text{in $\square\backslash R$ for all  $ \tilde{v} \in V(\t)$},
\]
which, along with  \eqref{15.6.24e2},  shows that $b$ (cf. \eqref{prebform}) can be re-arranged on $V(\t)$ as follows: 
\begin{equation}\label{bform}
b(v,\tilde{v}) : = \int_{\square } \nabla v_2 \cdot \overline{\nabla \tilde{v}_2} + \gamma^{-1}\int_\square \big( \div{v}\, \overline{\div{\tilde{v}}} + \curl{v}\, \overline{\curl{\tilde{v}}} \big), \quad \forall v,\tilde{v} \in V(\t).
\end{equation}
Furthermore, integration by parts gives
\begin{equation}\label{ibp}
	\int_{ \square} (  \div u \, \overline{ \div \phi} + \curl{u}\, \overline{\curl \phi}   )  = \int_\square \nabla u : \overline{\nabla \phi}, \qquad \forall u,\phi \in H^1_\t(\square;\CC^2).
\end{equation}
Thus, 
\begin{equation}\label{betapositive}
	\begin{aligned}
		b(v,\tilde{v}) &= \int_{\square}\nabla v_2   \cdot \overline{\nabla\tilde{v}_2}	+\gamma^{-1}  \nabla v : \overline{\nabla \tilde{v}}, \quad \forall v,\tilde{v} \in V(\t).
	\end{aligned}
\end{equation}
Additionally, notice that $c$ (cf. \eqref{cform}) is an equivalent $L^2(\square;\CC^2)$ inner product. Therefore, $b+c$ is an equivalent $H^1$-inner product on $V(\t)$.  
Let  $\H$ denote the closure of $V(\theta)$ in $L^2(\square;\CC^2)$, i.e.
\begin{equation*}
	\label{Ht}
	\mathcal{H} : = \{ h \in L^2(\square;\CC^2) \, | \, \div{h} = \curl{h} = 0 \text{ in $\square\backslash R$}  \},
\end{equation*}
 and consider  $\B(\t)$ be the non-negative self-adjoint operator in  $(\H,c)$ generated by the form $b$ with form domain $V(\t)$, i.e.
\begin{equation*}
	b(v, \tilde{v}) = c(\B(\t)v,\tilde{v}), \quad \forall\, v \in \mathrm{dom}(\B(\t)), \, \forall\, \tilde{v} \in V(\t).
\end{equation*}
Now,  \eqref{degp} can be rewritten as: Find $\lambda \in [0,\infty)$ such that there exists a non-trivial $v \in \mathrm{dom} (\B(\t))$ satisfying
\begin{equation}\label{opdegp}
	\B(\t) v = \lambda v.
\end{equation}

Notice that,  since $H^1(\square;\CC^2)$ is compactly embedded into $L^2(\square;\CC^2)$,  the resolvent of $\B(\theta)$ is  compact (cf. \eqref{betapositive}) and, in particular,  the spectrum  of $\B(\t)$ consists of countably many isolated non-negative eigenvalues of finite multiplicity which tend to infinity.

Upon considering all of the above, we have established that: there exists non-trivial waves propagating down a two-dimensional photonic crystal fibre, on the critical light line, $k=1$, if, and only if, $\lambda= \omega^2$ is an eigenvalue of $\B(\t)$ for some $\t \in \square^*$.\\

\textbf{Below the critical line: $k<1$.} It is clear that $a$ is bounded in $H^1$ and $b_\ep$ is uniformly bounded in $H^1$. Furthermore, we shall demonstrate below that $\ep^{-1} a+b_\ep$ is uniformly positive on $H^1$ in the following sense: there exists $\nu>0$ such that\footnote{Henceforth, $\beta[\cdot]: = \beta(\cdot,\cdot)$ for a sesquilinear form $\beta$.} 
\begin{equation}\label{epformlowerbound}
\ep^{-1} a[\phi] + b_\ep[\phi]\ge \nu \| \nabla \phi \|_{L^2(\square;\CC^2)}^2  
\end{equation}
for all $\phi \in H^1(\square;\CC^2)$ and all $\ep\in(0,1)$.
Therefore $\ep^{-1} a + b_\ep + c$ is an equivalent inner product on $H^1(\square;\CC^2)$  that is uniformly positive in $\ep$. Let us consider $\B_\ep(\t)$, the non-negative self-adjoint operator in  $(L^2(\square;\CC^2),c)$ generated by the form $\ep^{-1} a+b_\ep$ with form domain $H^1_\t(\square;\CC^2)$, i.e.
\begin{equation*}
	\ep^{-1}a(\phi, \psi) + b_\ep(\phi,\psi) = c(\B_\ep(\t)\phi,\psi), \quad \forall\, \phi\in \mathrm{dom}(\B_\ep(\t)), \, \forall\, \psi \in H^1_\t(\square;\CC^2).
\end{equation*}
Then problem \eqref{nondp} can be rewritten as: Find $\lambda \in [0,\infty)$ such that there exists a non-trivial  $u \in \mathrm{dom}(\B_\ep(\t))$ satisfying 
\begin{equation*}
\B_\e(\t) u = \lambda u.
\end{equation*}
As  $\B_\ep(\t)$ has a compact resolvent, then non-trivial waves propagate down the two-dimensional photonic crystal fibre, for $k =\sqrt{1-\ep}$, $0<\ep<1$, if, and only if, $\lambda= \omega^2$ is an eigenvalue of $\B_\ep(\t)$ for some $\t \in \square^*$.

\textit{Proof of \eqref{epformlowerbound}}
Notice that, from \eqref{15.6.24e2} and 
\begin{equation}\label{15.6.24e3}
	\big| \nabla^\perp u_1 \cdot \overline{\nabla u_2} + \nabla u_2 \cdot \overline{\nabla^\perp u_1}\big| \le 2 | \nabla u_1| |\nabla u_2|,
\end{equation}
we have
\begin{flalign*}
&\int_{\square \backslash R}\left( |\div\phi|^2 +|\curl \phi|^2 
+  (\sqrt{1-\ep}-1)\big(  \nabla^\perp \phi_1 \cdot \overline{\nabla \phi_2} + \nabla \phi_2 \cdot \overline{\nabla^\perp \phi_1}\big)\  \right) \\
& = \int_{\square \backslash R}\left(  |\nabla \phi|^2 + \sqrt{1-\ep}\big(  \nabla^\perp \phi_1 \cdot \overline{\nabla \phi_2} + \nabla \phi_2 \cdot \overline{\nabla^\perp \phi_1}\big)\  \right) \\
& \ge\int_{\square \backslash R}\left(  |\nabla \phi|^2 - 2\sqrt{1-\ep}|\nabla \phi_1||\nabla \phi_2|  \right)  \\
& \ge \int_{\square \backslash R} (1 - \sqrt{1-\ep}) |\nabla \phi|^2 \\
&\ge\frac{\ep}{2} \int_{\square \backslash R}  |\nabla \phi|^2.
\end{flalign*}
Similarly, for $0<\eta<1$ the positive solution to  $1-\eta = \epsilon_0 - \eta^{-1}$, one has 
\begin{flalign*}
	&\int_{ R}\left( |\nabla \phi_1|^2 +\epsilon_0|\nabla \phi_2|^2 
	+  \sqrt{1-\ep}\big(  \nabla^\perp \phi_1 \cdot \overline{\nabla \phi_2} + \nabla \phi_2 \cdot \overline{\nabla^\perp \phi_1}\big)\  \right) \\
	& \ge\int_{R}\left(  |\nabla \phi_1|^2 +\epsilon_0|\nabla \phi_2|^2  - 2\sqrt{1-\ep}|\nabla \phi_1||\nabla \phi_2|  \right)  \\
	& \ge \int_{R} \Big( (1 - \eta) |\nabla \phi_1|^2 + (\epsilon_0 - \eta^{-1}(1-\ep)) |\nabla \phi_2|^2\Big) \\
	& =\int_{R} \Big( (1 - \eta) |\nabla \phi_1|^2 + (1-\eta + \eta^{-1}\ep) |\nabla \phi_2|^2\Big) \\
	&\ge(1-\eta) \int_{R} \Big(|\nabla \phi_1|^2+|\nabla \phi_2|^2\Big)= (1-\eta) \int_{R} |\nabla \phi|^2 .
\end{flalign*}
Thus,  the above assertions (upon recalling  \eqref{aform}, \eqref{bepform} and $\gamma = \epsilon_0-1$) give 
\[
\ep^{-1}a[\phi] + b_\ep[\phi]  \ge \frac{1}{2} \| \nabla \phi\|^2_{L^2(\square\backslash R;\CC^2)}+ \frac{1}{\gamma +\ep} (1-\eta) \| \nabla \phi\|^2_{L^2(R;\CC^2)} \ge \nu\| \nabla \phi\|^2_{L^2(\square;\CC^2)}
\]
for $\nu = \min\{\tfrac{1}{2},\tfrac{1-\eta}{\gamma+1}\}$. \qed

\section{Comparison of spectral problems  in the vicinity of the critical line.}\label{s.resolventestimates}

In this section we shall provide uniform in $\t$, order $\ep$  error estimates between the resolvent of $\B_\ep(\t)$ and the pseudo-resolvent of $\B(\t)$ in operator norm. This in turn will provide  uniform, order $\ep$, error estimates between the spectra of $\B_\ep(\t)$ and $\B(\t)$.
\\

Consider $H^{-1}(\square;\CC^2)$ the dual of $H^{1}(\square;\CC^2)$ equipped with the standard norm. The following result is proved in  Appendix \ref{app2}:
\begin{lem}\label{prethm.resolventasymp}
Fix $\ep >0$, $\t \in \square^*$ and $F \in H^{-1}(\square;\CC^2)$. Let $U_{\ep,\t} \in H^1_\t(\square;\CC^2)$ solve
\[
\ep^{-1} a(U_{\ep,\t},\phi) + b_\ep(U_{\ep,\t},\phi)  +c(U_{\ep,\t},\phi) = \langle F, \phi \rangle, \quad \forall \phi \in H^1_\t(\square;\CC^2),
\]
and let $V_{\t} \in V(\t)$ solve
\[
 b(V_{\t},\tilde{V})  +c(V_{\t},\tilde{V}) = \langle F, \tilde{V} \rangle, \quad \forall\, \tilde{V} \in V(\t).
\]
Then, there exists a constant $C>0$, independent of $\ep$, $\t$ and $F$, such that 
\[
\| U_{\ep,\t}- V_\t \|_{H^1(\square;\CC^2)} \le C \ep \|F\|_{H^{-1}(\square;\CC^2)}.
\]
\end{lem}
Lemma \ref{prethm.resolventasymp} can be rewritten into operator theoretic language in terms of $H^{-1}(\square;\CC^2)$ to $H^1(\square;\CC^2)$ solution operators and projections onto the appropriate annihilator sets. However, we shall simply state here that this result implies, amongst other things, resolvent estimates between $\mathcal{B}_\ep(\t)$ and $\mathcal{B}(\t)$. Indeed, for fixed $f\in L^2(\square;\CC^2)$,  consider $F\in H^{-1}(\square;\CC^2)$ given by the action   $\langle F,\phi \rangle := c(f,\phi)$; moreover, let $P : L^2(\square;\CC^2) \rightarrow \mathcal{H}$ be the orthogonal projection with respect to the inner product $c$. Then,  Lemma \ref{prethm.resolventasymp} immediately implies the following:
\begin{lem}\label{thm.resolventasymp}
	There exists a constant $C>0$, independent of $\ep$ and $\t$ such that  
	\[
	\| (\B_\ep(\t)+I)^{-1} -(\B(\theta)+I)^{-1} P  \|_{L^2(\square;\CC^2) \rightarrow L^2(\square;\CC^2)} \le C \ep, \quad \forall \t \in \square^*, \forall 0<\ep<1.
	\]
\end{lem}

Lemma \ref{thm.resolventasymp}, and the fact that  $V(\t)$ is infinite-dimensional, readily provides a uniform asymptotic comparision on spectra(see \cite[Section 6]{SCIKVPS23} for details). Indeed, let
\[
0 \le \lambda_{1,\ep} (\theta) \le \lambda_{2,\ep} (\theta) \le \lambda_{3,\ep} (\theta) \le \ldots, \qquad \lambda_{n,\ep}(\theta) \rightarrow \infty \ \text{as} \ n\rightarrow \infty.
\]
and 
\[
0 \le \lambda_1 (\theta) \le \lambda_2 (\theta) \le \lambda_3 (\theta) \le \ldots, \qquad \lambda_n(\theta) \rightarrow \infty \ \text{as} \ n\rightarrow \infty.
\]
denote  the eigenvalues of the operators $\B_\ep(\t)$ and $\B(\t)$ respectively, ordered according to the min-max principle.  Then:
\begin{lem} There exists a constant $C>0$ such that 
	\[
	\left|\frac{1}{\lambda_{n,\ep}(\t)+1} - \frac{1}{\lambda_n(\t)+1}\right| \le C \ep, \qquad \forall\, n \in \mathbb{N}, \forall\, \t \in \square^*, \forall\, 0<\ep<1.
	\]
\end{lem}
This, in turn, provides a quantitative asymptotic description  of the closeness of the Maxwell spectrum on the line $k=\sqrt{1-\ep}$,
\[
\sigma_\ep = \overline{\bigcup_{\t \in \square^*} \sum_{n=1}^\infty \lambda_{n,\ep}(\t)},
\] to the Maxwell spectrum on the critical line $k=1$,
\begin{equation*}
	\sigma_0 := \overline{\bigcup_{\t \in \square^*} \sum_{n=1}^\infty \lambda_n(\t)}.
\end{equation*}
Indeed, the following result holds (see, for example, \cite[Section 6]{SCIKVPS23} for details):
\begin{lem}\label{thmSpecFell}
	For every compact $[a,b]$ in $\RR$, 
	one has
	\[
	d_{[a,b]} \left( \sigma_\ep, \sigma_0\right) \le C_b \ep,
	\]
	for some positive constant $C_b$ independent of $\ep$ and $a$.  Here,
	\[
	d_{[a,b]}(X,Y) : = \max\{{\rm{dist}([a,b] \cap X, Y)}, \rm{dist}( [a,b]\cap Y,X) \}
	\]
	where $\rm{dist}(X,Y) := \sup_{x\in  X} \rm{dist}(x,Y)$ is the non-symmetric distance of  set $X$ from set $Y$ (adopting the usual convention that $\rm{dist}(\emptyset,A)=\rm{dist}(A,\emptyset) =0$ for any set $A$).
\end{lem}
\begin{rem}
Lemma \ref{thmSpecFell}, establishes convergence of $\sigma_\ep$ to $\sigma_0$ in the Fell topology with rate. 
\end{rem}

Finally, the  spectrum $\sigma_0$ can be represented as (see Corollary \ref{cor.cont} below)
\[
\sigma_0  = \overline{\sum_{n=1}^\infty [a_n,b_n]} \quad a_n : = \min_{\t\in\square^*} \lambda_n(\t), b_n : = \max_{\t\in\square^*} \lambda_n(\t),
\]
and we see that $\sigma_0$ has a gap if two adjacent spectral bands do not overlap: i.e. if for some $m \in \mathbb{N}$, one has 
\[
b_m < a_{m+1}.
\]
An immediate consequence of Lemma \ref{thmSpecFell} is the following result (which is a main result from the perspective of wave propagation in two-dimensional photonic crystal fibres):
\begin{thm}\label{thm.PBG}
	Assume that, for some $m\in \NN$, $I=(b_m,a_{m+1})$ is a gap of $\sigma_0$ . Let $\ep_0 := (a_{m+1}-b_m)/2C_b$.  Then, the set
	\[
\G = \bigcup_{0\le \ep < \min\{1, \ep_0\}} \{ (\omega, \sqrt{1-\ep}) \, | \, \omega^2 \in (b_m+C_b \ep, a_{m+1} - C_b\ep) \},
\]
is a photonic band gap. That is, for any pair $(\omega, k) \in \G$, the Maxwell system  \eqref{coef},\eqref{Maxwellsystem} does not admit a non-trivial solution of the form \eqref{propsol} for any $\t \in \RR^2$. 
\end{thm}
\begin{proof}
Since  $\ep < \ep_0$, one has  $b_m+C_b \ep < a_{m+1} - C_b \ep$. Then, Lemma \ref{thmSpecFell} implies that $(b_m+C_b \ep, a_{m+1}- C_b\ep)$ is a gap in the spectrum of $\sigma_\ep$.
	Thus, $\G$ 	satisfies the statements of the corollary. 
\end{proof}
\section{Some properties of the reduced Maxwell operators $\B(\t)$}
\label{secblochlim}
In this section we shall study the dependence of the operators $\B(\t)$ and, subsequently, of their spectral properties, with respect to the quasi-momenta $\t$.

An important distinct feature of the operators $\B(\theta)$ (as opposed to the usual Floquet-Bloch operators concerned with periodic elliptic PDEs in the whole space) is the non-trivial $\t$-dependence of the  form domain $V(\theta)$. The statement of continuity of $\B(\theta)$, and subsequently its eigenvalues $\lambda_n(\t)$, is therefore not a simple consequence of the Bloch-wave representation of functions belonging to $H^1_\theta(\square;\CC^2)$. Such statements rely upon establishing that the underlying space $V(\theta)$ is continuous with respect to $\theta$.

In the rest of this section we shall introduce and prove the results alluded to above.
\begin{thm}[Lipschitz continuity of $V(\theta)$]
	\label{lem:spcom1}
	There exists a constant $L_V >0$ such that
	\begin{equation*}
\forall \t, \t' \in \square^*, \forall v \in V(\t), \quad \inf_{v'\in V(\t')}\|v-v'\|_{H^1(\square;\CC^2)} \le L_V |\t-\t'| \| v \|_{H^1(\square;\CC^2)}
	\end{equation*}
\end{thm}

To prove this result we shall use the following characterisation of the family of spaces $V(\theta)$.

\begin{lem}
	\label{lem:spcom2} $ $ \\
	\noindent (i). Suppose $\theta \in \square^*\backslash\{0\}$, then $v \in V(\theta)$ if, and only if,
	\[
	v = \nabla a + \nabla^{\perp}  b 
	\]
	for some $a,b \in H^2_{\theta}(\square): = \{e^{\i \t \cdot y} v \, | \, v \in H^2_{per}(\square) \}$ with
	\begin{alignat*}{3}
		\Delta a  = f_1, & \hspace{3cm} & \Delta b  = f_2,
	\end{alignat*}
for  some $f_1, f_2 \in L^2(\square)$ satisfying  $\mathrm{supp}{f_1}$ and $\mathrm{supp}{f_2}$ contained in $\overline{R}$. 

	\noindent (ii). Suppose $\theta = 0$. Then $v \in V(0)$ if, and only if,
	\[
	v = c + \nabla a + \nabla^{\perp} b 
	\]
	for some constant $c  \in \CC^2$ and $a,b \in H^2_{per}(\square)$ with\footnote{Henceforth, $\mv{f} := \int_\square f$.}   $\mv{a}=\mv{b}=0$ and 
	\begin{alignat*}{3}
		\Delta a  = f_1, & \hspace{3cm} & \Delta b  = f_2,
	\end{alignat*}
for  some $f_j  \in L^2(\square)$, $j=1,2$, satisfying   $\mv{f_j}=0$ and $\mathrm{supp}{f_j}$ contained in $\overline{R}$. 
\end{lem}

\begin{proof} In both cases (i) and (ii) the necessity is easy to demonstrate. Let us now show  sufficiency. 
	
	Fix $\theta \in \square$,  $v \in V(\theta)$ and set $f_1 : = \div v$, $f_2 : = \curl v$. It is clear that $f_1,f_2 \in L^2(\square)$ and that their supports are contained in $\overline{R}$. Now, let us introduce $a,b \in H^2_\theta (\square)$ the solutions of 
	\begin{alignat*}{3}
		\Delta a  = f_1, & \hspace{3cm} & \Delta b  = f_2.
	\end{alignat*}
	For $\t=0$, note $\langle f_i \rangle = 0$, and so there exist solutions $a,b$ with $\langle a \rangle = \langle b \rangle = 0$.
	Now,   $w : = v - \nabla a - \nabla^{\perp} b$ belongs to $H^1_{\theta}(\square;\CC^2)$ with $\div w =\curl w = 0$, which implies $w$ is  constant (see \eqref{ibp}). Finally, for $\t \neq 0$  the only constant function in $H^1_\t(\square;\CC^2)$ is the zero vector, i.e. $w \equiv 0$.
\end{proof}
Additionally, in the proof of Theorem \ref{lem:spcom1} we shall utilise the following uniform estimates.
\begin{lem} 
	\label{lem:spcom3}
	\hspace{0pt}
	
	\begin{enumerate}[(i)]
		\item{Let $\t \in \square^*$ and $p \in H^2_{\t}(\square)$ satisfy
			\[
\int_\square		\Delta p e^{-\i \t \cdot y} = 0.
			\]
		Then, there exists a constant $C>0$, independent of $p$ and $\t$, such that 
	\begin{equation}\label{pH2}
		\| p \|_{H^2(\square)} \le C\| \Delta p \|_{L^2(\square).}
	\end{equation}

		}
		\item{Suppose $0 \neq \theta \in \square^*$, $f \in L^2(\square)$ and consider $ q \in H^2_{\t}(\square)$ the solution to

			\begin{equation}\label{eq:spcom8.1}
		\Delta q = e^{\i \t \cdot y}f\in L^2(\square).
			\end{equation}
					Then, there exists a constant $C>0$, independent of $p$ and $\t$, such that
			\begin{equation}\label{2decdesired}
				\left\|\nabla q +  \mv{f}\frac{\i \t}{|\t|^2} e^{\i \t\cdot y} \right\|_{H^1(\square;\CC^2)} \le C\norm{f}_{L^2(\square)}.			
			\end{equation}	
			}
	\end{enumerate}
\end{lem}
\begin{proof} \hspace{0pt}
	
	{\it Proof of (i):}
	 By Fourier series decomposition, we have 
	\begin{alignat*}{3}
		p(y)  = \sum_{z \in \mathbb{Z}^2} a_z  e^{ {\rm i}(\theta +  z) \cdot y}, & \hspace{3cm} & \Delta p(y)  =  \sum_{z \in \ZZ^2} b_z  e^{ {\rm i}(\theta +  z) \cdot y},
	\end{alignat*}
and we shall consider the equivalent $H^2$ norm
\[
\norm{p}:= \left( \frac{1}{2\pi}\sum_{z \in \mathbb{Z}^2} \left( 1 + \vert \theta +  z \vert^2 \right)^2 \vert a_z \vert^2 \right)^{1/2}.
\]
By assumption, $b_0 =0$. Equating coefficients gives $a_0 = 0$ and  $a_z = -\tfrac{b_z}{\vert \theta + z \vert^2}$, $z \neq 0$. Now, since $\t \in \square^* = [-\frac{1}{2},\frac{1}{2})^2$,
\begin{equation*}
\begin{aligned}
\norm{p}^2 &= \frac{1}{2\pi}\sum_{z \in \mathbb{Z}^2} \left( 1 + \vert \theta +  z \vert^2 \right)^2 \vert a_z \vert^2 = \frac{1}{2\pi}\sum_{0 \neq z \in \mathbb{Z}^2} \Big( \frac{1 + \vert\theta +  z \vert^2}{\vert \theta +  z \vert^2} \Big)^2 \vert b_z \vert^2 \\
& \le \Big(2+1\Big)^2 \frac{1}{2\pi}\sum_{0\neq z \in \mathbb{Z}^2}\vert b_z \vert^2 = 9  \| \Delta p \|_{L^2(\square)}^2.
\end{aligned}
\end{equation*}

	{\it Proof of (ii):} 
%
Note that $f = \langle f \rangle + f_0$ where $f_0 \in L^2(\square)$ satisfies $\langle f_0 \rangle =0$. Thus $q = -\frac{\langle f \rangle}{|\t|^2} e^{\i \t \cdot y} + p$ where 
\[
p \in H^2_\t(\square) \text{ solves $\Delta p = e^{\i \t \cdot y}f_0$.}
\] Note that $\int_\square \Delta p e^{-\i \t \cdot y} =0$ and so, \eqref{pH2}holds for $p$. Furthermore,
\[
\nabla q = -\i \t \frac{\langle f \rangle}{|\t|^2} e^{\i \t \cdot y} + \nabla p.
\]
Hence,   \eqref{2decdesired} holds. 
\end{proof}

\begin{proof}[Proof of Theorem \ref{lem:spcom1}.] 	 Fix $\t,\t' \in \square^*$. 
	For each $v \in V(\t)$, one needs to construct a particular $v' \in V(\t')$ such that 
	\begin{equation}\label{tlip}
		\|v-v'\|_{H^1(\square;\CC^2)} \le L_V |\t-\t'| \| v \|_{H^1(\square;\CC^2)}.
	\end{equation}
We begin by observing that  it is sufficient to consider $v\in V(\t)$ that satisfy
\begin{equation}\label{toshow15dec}
\mv{e^{-\i \t \cdot y}\div{v}}= 0 \quad \text{and}\quad \mv{e^{-\i \t \cdot y}\curl{v}}=0.
\end{equation}
Indeed, for a fixed $\eta \in C^\infty_0(R)$ such that $\int_R \eta = 1$, then for any  $V \in V(\t)$, one can verify that
\[
	v: =V - \eta e^{\i \t \cdot y} \mv{e^{-\i \t \cdot y}V}
\]
satisfies \eqref{toshow15dec} and that $\Psi: = \eta e^{\i \t \cdot y} \mv{e^{-\i \t \cdot y}V} $ need not be approximated; indeed $\Psi$ belongs to $V(\t')$ for all $\t' \in \square^*$ (as $\mathrm{supp}\, \eta \in \overline{R }$).

	It remains to find $v' \in V(\t')$ such that \eqref{tlip} holds for any  $v\in V(\t)$ that satisfy  \eqref{toshow15dec}.

	{\it Case 1: } Here, we consider the case when $v = \nabla a + \nabla^\perp b$ for $a,b\in H^2_\t(\square)$ that satisfy
	\[
	\Delta a = f_1 \quad \text{and} \quad \Delta b = f_2 \quad \text{in $\square$};
	\]
	and, for $\t = 0$,  additionally satisfy $\langle a \rangle = \langle b \rangle =0$. Notice that, $f_1 = \div v$ and $f_2 = \curl v$, and so, by \eqref{toshow15dec}, the assumptions of  Lemma \ref{lem:spcom3} (i) hold; namely,  
	\begin{equation}\label{15dec2}
		\| a \|_{H^2(\square)}  + \| b \|_{H^2(\square)} \le C \left( \| \div{v}\|_{L^2(\square;\CC^2)} +  \|\curl{v} \|_{L^2(\square;\CC^2)} \right) \le C\sqrt{2} \| v\|_{H^1(\square;\CC^2)}.
	\end{equation}
	Let us consider $v' = \nabla a' + \nabla^\perp b'$ where, for $p=a$ or $b$, we set
	\begin{equation*}
		\begin{aligned}
			\Delta p' = e^{\i (\t' - \t)\cdot y} \div{v} + \chi_R e^{\i (\t'-\t) \cdot y} \Big(  2\i (\t' - \t) \cdot \nabla p-|\t' - \t|^2 p\Big) + \chi_R e^{\i \t' \cdot y} \frac{D_p}{|R|}, \\
		\end{aligned}
	\end{equation*}
	where $\chi_R$ is the characteristic function of the inclusion $R$ and the constants $D_p$ will be chosen below. 	By  Lemma \ref{lem:spcom2}, $v' \in V(\t')$ and we shall now show that \eqref{tlip} holds.
	
	  Note that, one has
	 \begin{flalign*}
	 	 \nabla p - \nabla p'&= \nabla p - \nabla (e^{\i (\t'-\t)\cdot y} p)+\nabla (e^{\i (\t'-\t)\cdot y} p - p')
	 \end{flalign*}
	 and, by the inequalities $\|1 - e^{\i \Theta \cdot y}\|_{L^\infty(\square)} \le |\Theta|$ and \eqref{15dec2}, one has 
	\[
	\begin{aligned}
	\| \nabla p - \nabla (e^{\i (\t'-\t)\cdot y} p) \|_{H^1(\square;\CC^2)} \le K | \t' - \t | 	\| p \|_{H^2(\square)} \le K | \t' - \t | 	\| v\|_{H^1(\square;\CC^2)}, \\
	\end{aligned}
	\]	
for some $K$ independent of $v$, $\t$ and $\t'$.	So it remains to appropriately estimate the $\t'$-quasiperiodic functions  $r_p : = p' - e^{\i (\t'-\t)\cdot y} p$. Now
	\[
	\begin{aligned}
\Delta r_p  &= F_p : = (\chi_R- 1) e^{\i (\t'-\t) \cdot y} \Big(  2\i (\t' - \t) \cdot \nabla p-|\t' - \t|^2 p \Big) + \chi_R e^{\i \t' \cdot y} \frac{D_p}{|R|}, 
	\end{aligned}
	\]
At this point, we fix $D_p$ such that $\int_\square F_p e^{-\i \t' \cdot y} = 0$, i.e.
	\begin{equation*}
		\begin{aligned}
			D_p =\int_{\square\backslash R} e^{-\i \t \cdot y} \big(  2\i (\t' - \t) \cdot \nabla p - |\t' - \t|^2 p \big) , 
		\end{aligned}
	\end{equation*}
Then, by Lemma \ref{lem:spcom3}(i)  and \eqref{15dec2} there exists a constant $K>0$ independent of $v$, $\t$ and $\t'$ such that

\begin{equation*}
\| r_a \|_{H^2(\square)}   + \| r_b \|_{H^2(\square)} \le K|\t - \t'| \| v\|_{H^1(\square;\CC^2)}.
\end{equation*}
Hence \eqref{tlip} holds.

{\it Case 2:} It remains, to consider the case when $\t=0$ and $v = c \in \CC^2$.  One needs to show that for every $0 \neq \t' \in \square^*$, there  exists $v \in V(\t')$ such that 
\begin{equation}\label{0lip}
\| c - v\|_{H^1(\square;\CC^2)} \le K | \t' | |c |, 
\end{equation}
for some constant $K>0$ independent of $\t'$ and $c$. 


For fixed $c \in \CC^2$, $\t' \neq 0$, note 
\[
c =-(\i \t' \cdot c)\frac{ \i \t'}{|\t'|^2}  - (\i \t'^\perp \cdot c)\frac{\i \t'^\perp}{|\t'|^2} .
\]
Let us consider  $p,q \in H^1_{per}(\t')$ such that 
\begin{equation*}
\begin{aligned}
\Delta p = (\i \t' \cdot c ) e^{\i \t' \cdot y} \frac{\chi_R}{|R|}\quad \& \quad 
 \Delta q =  (\i \t'^\perp \cdot c) e^{\i \t' \cdot y} \frac{\chi_R}{|R|}.
\end{aligned}
\end{equation*}
 By Lemma \ref{lem:spcom3} (ii) one has 
\[
	\left\|\nabla p + e^{\i \t' \cdot y}(\i \t'\cdot c)\frac{\i \t'}{|\t'|^2} \right\|_{H^1(\square)} \le \frac{C}{\sqrt{|R|}} |\t'|| c|, \quad \& \quad \left\|\nabla^\perp q + e^{\i \t' \cdot y}(\i \t'^\perp \cdot c)\frac{\i \t'^\perp}{|\t'|^2} \right\|_{H^1(\square)} \le \frac{C}{\sqrt{|R|}}  |\t'|| c|.
\]
That is  $v = \nabla p + \nabla^\perp q$ belongs to $V(\t')$ (see Lemma \ref{lem:spcom2}) and satisfies
\[
	\|v - e^{\i \t' \cdot y} c \|_{H^1(\square)} \le \frac{2C}{\sqrt{|R|}}  |\t'|| c|,
\]
Clearly, this implies \eqref{0lip}. The proof is complete. 
\end{proof}

Theorem \ref{lem:spcom1} implies Lipschitz continuity of the operator $\B(\t)$ in an appropriate sense: 
\begin{lem}
	\label{apb.1}
	
	For any  $\theta, \theta'\in \square^*$ and $F\in H^{-1}(\square;\CC^2)$, consider the solutions $v_\t \in V(\t)$, $v_{\t'} \in V(\t')$ to 
	\begin{align}
		b(v_\t,\tilde{v}) + c(v_\t, \tilde{v})& = \langle F,\tilde{v}\rangle , \quad \forall \tilde{v} \in V(\t);\label{6dec1}\\
		b(v_{\t'},\tilde{v}') + c(v_{\t'}, \tilde{v}') &= \langle F,\tilde{v}' \rangle, \quad \forall \tilde{v}' \in V(\t'); \label{6dec2}
	\end{align}

There, exists a constant $K>0$ independent of $\t,\t',F$ such that
\[
\| v_\t - v_{\t'} \|_{H^1(\square;\CC^2)} \le K | \t - \t'| \| F  \|_{H^{-1}(\square;\CC^2)}. \qedhere
\]
\end{lem}
 Lemma \ref{apb.1}  implies
\[
	\|(\B(\t) +I)^{-1}- (\B(\t')+I)^{-1}\|_{\H \rightarrow \H} \le K| \t - \t'|.
\]
This in turn, implies, by standard arguments, that the spectral band functions, $\t \mapsto \lambda_n(\t)$, of $\B(\t)$ are continuous in $\t$: 
\begin{cor}\label{cor.cont}
For each $n\in \NN$, the eigenvalues $\lambda_n(\t)$ of $\B(\t)$ are continuous functions of $\t$ on $\square^*$. 
\end{cor}
\begin{proof}[Proof of Lemma \ref{apb.1}]

Consider the $H^1(\square;\CC^2)$ equivalent inner product
\begin{equation*}
	Q(u,\phi) : = \int_{\square}\nabla u_2\cdot \overline{\nabla \phi_2} + \gamma^{-1} \nabla u : \overline{\nabla \phi} + \int_\square u \cdot  \overline{\phi} + \gamma \int_R u_2 \overline{\phi_2},
\end{equation*}
and denote the associated equivalent $H^1$-norm by $\| \cdot \|_{+}$. Let $\| \cdot \|_{-1}$ be the dual norm with respect to the $\| \cdot \|_{-}$  norm.  For each $\t \in \square^*$, consider $\mathcal{P}(\t): H^1(\square;\CC^2) \rightarrow V(\t)$ the orthogonal projection with respect to the $\| \cdot \|_{+}$ norm.

Recall $b+c$  equals $Q$ on $V(\t)$ (see \eqref{cform} and \eqref{betapositive}).Then it follows that
  \begin{equation}
Q(v_\t - \mathcal{P}(\t)v_{\t'},\tilde{v})  = \langle F,(I- \mathcal{P}(\t')\tilde{v}\rangle , \quad \forall \tilde{v} \in V(\t);\label{6dec3}
  \end{equation}
Indeed, by \eqref{6dec1}, \eqref{6dec2}, and the identity $Q(\mathcal{P}(\t)v_{\t'},\tilde{v})  =Q( v_{\t'},\mathcal{P}(\t')\tilde{v}) $ it follows that
\begin{flalign*}
Q(v_\t - \mathcal{P}(\t)v_{\t'},\tilde{v})  & = \langle F, \tilde{v} \rangle - Q( \mathcal{P}(\t)v_{\t'},\tilde{v})  = \langle F, \tilde{v} \rangle - Q( v_{\t'},\mathcal{P}(\t')\tilde{v})  \\
& = \langle F, \tilde{v} \rangle  - \langle F, \mathcal{P}(\t')\tilde{v} \rangle .
\end{flalign*}
Substituting $\tilde{v}= v_\t - \mathcal{P}(\t) v_{\t'}$ in \eqref{6dec3},
gives
\begin{equation}\label{6dec4}
\| v_\t - \mathcal{P}(\t) v_{\t'} \|_{+}^2  \le \| F \|_{-} \|  (I- \mathcal{P}(\t')) \big(  v_\t - \mathcal{P}(\t) v_{\t'} \big)  \|_{+}.
\end{equation}
Now,  Theorem \ref{lem:spcom1} implies that 
\begin{equation}\label{6dec5}
\| \big(I- \mathcal{P}(\t')\big)\tilde{v} \|_{+} \le C| \t - \t| \| \tilde{v} \|_{+}, \quad \forall \tilde{v} \in V(\t), \forall \t \in \square^*,
\end{equation}
for some constant $C$. Therefore, \eqref{6dec4}, \eqref{6dec5} 
 imply
\[
 \| v_\t - \mathcal{P}(\t) v_{\t'} \|_{+} \le C | \t - \t'| \| F  \|_{-}.   
\]
This inequality, along with another application of \eqref{6dec5} and the standard bound 
\[
\| v_{\t'} \|_{+} \le \| F \|_{-},
\]
proves the desired result.
\end{proof}



\section{Photonic crystals which exhibit gaps in the Bloch spectrum on the critical light-line}
\label{examples}

In this section we shall demonstrate gaps in the limit spectrum 
\[
\sigma_0  = \overline{\bigcup_{n=1}^\infty [a_n,b_n]}, \quad a_n : = \min_{\t\in\square^*} \lambda_n(\t), b_n : = \max_{\t\in\square^*} \lambda_n(\t).
\]
\subsection{One-dimensional photonic crystal fibres}\label{s.1d}
Let us suppose we restrict ourselves to one-dimensional photonic crystal fibres, that is our dielectric media has additional homogeneous axis, say the $x_2$-axis. That is, $R$ can be assumed to be of the form $[-\pi,-\pi+h)\times[-\pi,\pi)$, for some $0<h<\pi$; i.e. the dielectric permittivity $\epsilon$ (see \eqref{coef}) on the periodic cell $\square$ is of the form
\[
\epsilon(x_1,x_2,x_3) = \left\{ \begin{matrix}	\epsilon_0,  &  x_1\in [-\pi,-\pi+h) \\[0.1em]
	1, & x_1\in[-\pi+h,\pi).
\end{matrix} \right.
\] 
Let us rewrite the spectral problem on the critical line in this case (cf.  \eqref{degp} and \eqref{opdegp}).  Let, $\lambda = \omega^2$ be an eigenvalue of $\mathcal{B}(\t)$, for some $\t \in \square^*$, with eigenfunction $u\in V(\t)$.
Without loss of generality $\t = (\Theta,0)$, $\Theta \in [-\frac{1}{2}, \frac{1}{2})$ and $u$ depends only on $x_1$, i.e. $u(x_1,x_2) \equiv u(x_1)$.\footnote{
In general, we are considering plane EM-waves   of the form (see \eqref{propsol})
 	\[
	\tilde{E}(x) =  E(x_1) e^{i (k_2 x_2 + \omega k x_3 - \omega t)  }, \quad \text{and} \quad  \tilde{H}(x) = H(x_1) e^{i (k_2 x_2 + \omega k x_3 - \omega t)  },
	\]
for $\t$-quasi-periodic amplitudes, $E,H$, with $\t \in [-\frac{1}{2},\frac{1}{2})$. However,  we can suppose without loss of generality $k_2 = 0$; indeed, for the general wave number $(k_2,\omega k)$,  we can always rotate the one-dimensional PCF so that $(k_2,\omega k)$ is an homogeneous axis of the media (that we would denote as the third axis). Thus, we are propagating down one homogeneous axis but not the other.}. 
Notice that if $\phi \in V(\t)$  and $\phi(x_1,x_2) \equiv \phi(x_1)$ then  $\phi$ is a constant vector for  $x_1 \in [h,\pi)$. Indeed, here $\div{\phi} = \phi_1'$ and $\curl{\phi} = \phi_2'$, then $\phi$ is constant in $(h,\pi+h)$ and
\[
\frac{1}{2\pi}b(u,\phi) =  \int_{-\pi}^{-\pi+h} \Big( \gamma^{-1}u_1' \overline{\phi_1'}(x_1) +(1+ \gamma^{-1}) u_2'(x_1) \overline{\phi_2'}(x_1) \Big) \, \mathrm{d}x_1,
\]
\[
\frac{1}{2\pi}c(u,\phi) =\int_{-\pi}^{-\pi+h} \Big( u_1(x_1) \overline{\phi_1}(x_1) +(1+ \gamma) u_2(x_1) \overline{\phi_2}(x_1)  \Big)\, \mathrm{d}x_1  + (2\pi-h) u(-\pi+h)  \overline{\phi(-\pi+h)}.
\]
see \eqref{bform} and \eqref{cform}.  Now, let us find the strong form of  \eqref{degp}; taking test functions of the form $\phi \in C_0^\infty(-\pi,-\pi+h)$ in \eqref {degp} gives
\[
\int_{-\pi}^{-\pi+h} \Big( \gamma^{-1}u_1'' \overline{\phi_1} +(1+ \gamma^{-1}) u_2'' \overline{\phi_2} \Big) \, \mathrm{d}x_1  = \lambda \int_{-\pi}^{-\pi+h} \Big( u_1 \overline{\phi_1} +(1+ \gamma) u_2 \overline{\phi_2}  \Big)\, \mathrm{d}x_1. 
\]
Thus,
\[
-u_1'' = \lambda \gamma u_1 \quad \text{in $(-\pi,-\pi+h)$}, \quad \text{and} \quad
-u_2''= \lambda \gamma u_2\quad \text{in $(-\pi,-\pi+h)$}.
\]
Then, from \eqref{degp}, the above equations, integration by parts in form $b$, and the fact $u$, $\phi$ are constant in $(-\pi+h,\pi)$, gives
\begin{equation}\label{4.06e1}
\Big( \gamma^{-1}u_1' \overline{\phi_1} +(1+ \gamma^{-1}) u_2' \overline{\phi_2}(x_1) \Big) \Big|_{-\pi}^{-\pi+h} = \lambda (2\pi-h) u(-\pi+h) \overline{\phi(-\pi+h)}.
 \end{equation}
Now, by $\Theta$-quasi-periodic and the fact $\phi$ is constant in $(-\pi+h,\pi)$, one has 
\[
\phi(-\pi+ h) =e^{\i \pi \Theta}  \phi(-\pi), 
\]
and likewise for $u$. Thus, \eqref{4.06e1} is equivalent to 
\[
  e^{-\i \pi \Theta}u_1' (-\pi+h) - u_1'(-\pi) = (\pi-h) \lambda \gamma u_1(-\pi),
\]
and 
\[
  e^{-\i \pi \Theta}u_2' (-\pi+h) - u_2'(-\pi)  = \frac{(\pi-h)}{1+\gamma} \lambda\gamma u_2(-\pi).
\]
That is, upon setting $\Lambda = \lambda \gamma$ , $l_1 = -\pi$, and $l_2 = -\pi+h$, then \eqref{degp}( equivalently \eqref{opdegp}) , for one-dimensional PCF, can be rewritten as the diagonal system on the inclusion phase $(-\pi,-\pi+h)$:
\begin{equation}\label{TE}\tag{TE} \left.
\begin{aligned}
-u_1'' = \Lambda  u_1 &\qquad \text{in $(l_1,l_2)$}, \\
 e^{-\i \pi \Theta}u_1' (l_2) - u_1'(l_1) &= (2\pi-h) \Lambda u_1(l_1), \\  
  e^{-\i \pi \Theta}u_1 (l_2)  &= u_1(l_1);
\end{aligned} \qquad \right\}
\end{equation}
and
\begin{equation}\label{TM}\tag{TM} \left.
	\begin{aligned}
		-u_2'' = \Lambda  u_2 &\qquad \text{in $(l_1,l_2)$}, \\
		e^{-\i \pi \Theta}u_2' (l_2) - u_2'(l_1) &= \frac{(2\pi-h)}{1+\gamma} \Lambda u_2(l_1), \\   
		e^{-\i \pi \Theta}u_2 (l_2)  &= u_2(l_1).
	\end{aligned} \qquad \right\}
\end{equation}
\begin{rem}
If \eqref{TE} has a solution then $u=(u_1,0)$ corresponds to the so-called Transverse-Electric polarisation for the one-dimensional PCF (with the $(x_1,x_3)$-plane being the `plane of incidence'); similarly, If \eqref{TM} has a solution then $u=(0,u_1,)$ corresponds to the so-called Transverse-Magnetic polarisation. Indeed, recall that \eqref{opdegp} is equivalent to the Maxwell equations \eqref{Max1}-\eqref{Max2} for $\lambda = \omega^2$ and $u=(H_3,E_3)$. 
\end{rem}
Notice that \eqref{TE} and \eqref{TM} are both problems of the same type; this type of problem also appears as the asymototic limit of the one-dimensional high-contrast double porosity model studied by one of the authors in \cite{Co1} and \cite{Co2}.  In particular, it was shown in \cite[Theorem 3.1]{Co2} that the spectral band functions $E_n(\Theta)$ of the associated operator has the usual Floquet-type properties; namely:
\begin{itemize}
\item{The functions $E_n$ are continuous and even around $\Theta = 0$;}
\item{For $n$ odd, $E_n$ is strictly increasing in $\Theta$ from $0$ to $\frac{1}{2}$};
\item{For $n$ even, $E_n$ is strictly decreasing in $\Theta$ from $0$ to $\frac{1}{2}$;}
\item{Denoting by $E_n^{(TE)}$, $E_n^{(TM)}$ the spectral bands functions for problem \eqref{TE} and \eqref{TM} respectively. Then, 
	\[
	\sigma_0 = \left(\bigcup_{n} \mathrm{Ran} E^{(TE)}_n\right) \cup \left(\bigcup_{n} \mathrm{Ran} E^{(TM)}_n\right) ; 
	\]
	where, for $j\in \{\mathrm{TE},\mathrm{TM}\}$, $\mathrm{Ran} E_n^{(j)}$, the range of $E_n^{(j)}$, equals to $[E_n^{(j)}(0), E_n^{(j)}(\frac{1}{2})]$ for $n$ odd and equals to $[E_n^{(j)}(\frac{1}{2}), E_n^{(j)}(0)]$ for $n$ even.}
\end{itemize}
These properties imply  that spectral gaps in $\sigma_0$ appear only at the edges of the Brillouin zone and, as such, one need only consider periodic and anti-periodic Floquet-Bloch modes ($\Theta= 0$ and $\frac{1}{2}$) for the TE and TM modes (\eqref{TE} and \eqref{TM}). Combining this result with Theorem \ref{thm.PBG} implies that, for one-dimensional photonic crystal fibres,  photonic band gaps near the critical line will appear if, and only if, for both the TE and TM polarisations, one has 
\[
\Lambda_n^{(j)}(\tfrac{1}{2}) < \Lambda^{(j)}_{n+1}(\tfrac{1}{2}) \quad \text{for some odd $n$},
\]
or
\[
\Lambda^{(j)}_n(0) < \Lambda^{(j)}_{n+1}(0) \quad \text{for some even $n$}.
\]
for $j =$ TE and $j=$ TM. If a gap appears for only one, but not both, polarisation, i.e. $j=$ TE or $j=$ TM, then the PCF has a so-called `weak' or partial photonic gap.
 
\subsection{ARROW fibres: two-dimensional photonic crystals with inclusions have vanishing volume faction}\label{s.Arrow}
Here, we provide an example of a genuine two-dimensional PCF that exhibits band gaps near the critical light line.

We being by rewriting the spectral problem  \eqref{opdegp}, for a general fibre cross section $R$, in a form that is more convenient for what follows. For each $\t \neq 0$, let  $ A_\theta : L^2(R) \rightarrow H^2_\t(\square)$ denote the bounded linear operator whose action is  $A_\theta  f = p$ where $p\in H^2_\t(\square)$ is the unique solution to 
\[
-\Delta p = \chi_R f \quad \text{in $\square$.}
\]
Here $\chi_R:L^2(R)\rightarrow L^2(\RR^2)$ is the operator that extends the input function by zero into $\RR^2 \backslash R$.
From  \eqref{lem:spcom2}, it follows that $V(\theta)$ has the following representation in terms of $A_\t$:
\[
V(\t) = \{ - \nabla A_\t f_1 - \nabla^\perp A_\t f_2 \, | \,  f=(f_1, f_2) \in L^2(R;\CC^2) \}.
\]
As such,  the eigenvalues $\lambda_n(\t)$  of $\mathcal{B}(\t)$ can be reformulated as:
\begin{equation}\label{minmax}
	\lambda_{n}(\t) =  \min_{\substack{V \subset L^2(R;\CC^2), \\ \dim{V} = n}}\, \max_{ f\in V} \frac{\mathfrak{b}_{\t}[f]}{c_{\t}[f]};
\end{equation}
for
\begin{equation*}
\mathfrak{b}_{\t}(f, \tilde{f}) : = b\big(\nabla A_\t f_1 + \nabla^\perp A_\t f_2, \nabla A_\t \tilde{f}_1 + \nabla^\perp A_\t \tilde{f}_2\big) 
\end{equation*}
and
\begin{equation*}
c_{\t}(f, \tilde{f}) : =  c\big(\nabla A_\t f_1 + \nabla^\perp A_\t f_2, \nabla A_\t \tilde{f}_1 + \nabla^\perp A_\t \tilde{f}_2\big) .
\end{equation*}
Notice that (after some integration by parts in  \eqref{bform})
\begin{equation}\label{20.6.24e1}
	\begin{aligned}
		&		\mathfrak{b}_{\t}(f, \tilde{f})\\
		&= \gamma^{-1}\int_{R}  f \cdot \overline{ \tilde{f}} - \int_{R}\left[ \Big((A_\t f_1)_{,22} + (A_\t f_2)_{,12}\Big) \overline{\tilde{f}_1} + \Big((A_\t f_1)_{,21} + (A_\t f_2)_{,11}\Big) \overline{\tilde{f}_2} \right],
	\end{aligned}
\end{equation}
\begin{equation}\label{20.6.24e2}
	\begin{aligned}
		&c_{\t}(f, \tilde{f}) \\
		& = \int_{R}  \Big(f _1 \overline{A_\t\tilde{f_1}} +  f_2 \overline{A_\t\tilde{f}_2} \Big) + \gamma \int_{R} \Big((A_\t f_1)_{,2} +(A_\t f_2)_{,1}\Big) \overline{\Big((A_\t \tilde{f}_1)_{,2} +(A_\t \tilde{f}_2)_{,1}\Big)} .
	\end{aligned}
\end{equation}\\

Now, let us consider the particular example of a two-dimensional photonic fibre with small circular inclusions, i.e. $R = \delta \Omega$, for small parameter $0<\delta < \pi$ and $\Omega \subset B_1$, where, henceforth, $B_\alpha$ is the ball of radius $\alpha$ centred at the origin.  
Note in passing that such models are known in the physics literature as ARROW fibres, see e.g. \cite{ARROW}.\\

 We shall establish the small $\delta$ leading-order behaviour of the eigenvalues given by  \eqref{minmax}, denoted now by $\lambda^{(\delta)}_n(\t)$, with error estimates uniform in $\t$. This will, in turn, allow us to demonstrate the existence of gaps in $\sigma_0$ for sufficiently small $\delta$.

 Consider the linear operator $A$ with domain $L^2(\Omega)$ and action $A f=p$ where $p$ is the solution to 
\[
-\Delta p = \chi_\Omega f \quad \text{in $\RR^2$}, \quad p(x) =  -\frac{d}{2\pi} \ln{|x|} + o(1) \ \text{as $|x| \rightarrow + \infty$},
\]
for constant $d =  \int_\Omega f$. 
Consider the forms 
\begin{equation*}
	\begin{aligned}
		\mathfrak{b}(f, \tilde{f}) : = \gamma^{-1} \int_{\Omega}  f \cdot \overline{ \tilde{f}}- \int_{\Omega}\left[ \Big((A f_1)_{,22} + (A f_2)_{,12}\Big) \overline{\tilde{f}_1} + \Big((A f_1)_{,21} + (A f_2)_{,11}\Big) \overline{\tilde{f}_2} \right],
	\end{aligned}
\end{equation*}
\begin{equation*}
\mathfrak{c}_1(f, \tilde{f})	: =  \left( \int_\Omega f_1 \right)\overline{\left( \int_\Omega \tilde{f}_1  \right)} + \left( \int_\Omega f_2 \right)\overline{\left( \int_\Omega \tilde{f}_2  \right)},
\end{equation*}
and
\begin{equation*}
 \mathfrak{c}_2(f, \tilde{f}) :=\int_{\Omega} \Big( f _1 \overline{A\tilde{f_1}} +  f_2 \overline{A\tilde{f}_2} \Big)  \\
+ \gamma \int_{\Omega} \Big((A f_1)_{,2} +(A f_2)_{,1}\Big) \overline{\Big((A \tilde{f}_1)_{,2} +(A \tilde{f}_2)_{,1}\Big)}.
\end{equation*}
Let 
\[
\nu_{\t,\delta} : =  g_\t(0)-\frac{1}{2\pi}\ln{\delta} ,
\]
where $g_\t$ is the regular part of  the Green's function of the $\theta$ quasi-periodic  Laplacian; i.e.
\begin{equation}\label{regG}
g_\t(x) : = G(x)  + \frac{1}{2\pi} \ln{|x|},
\end{equation}
where  $G$ is the $\theta$ quasi-periodic solution to 
\[
-\Delta G = \delta \quad \text{in $\square$}, \qquad \text{for $\delta$ the Dirac delta distribution.}
\]
Finally, consider
\begin{equation}\label{minmax2}
	\mu^{(\delta)}_{n}(\t) :=  \min_{\substack{V \subset L^2(\Omega;\CC^2), \\ \dim{V} = n}}\, \max_{ f\in V} \frac{\mathfrak{b}[f]}{\nu_{\t,\delta}\mathfrak{c}_1[f]+\mathfrak{c}_2[f]}.
\end{equation}
In Appendix \ref{App:Arrow}, we establish the following result:
\begin{thm}\label{thmArrowAsymptotics} There exists a constant $C_1>0$ such that 
	\[
	\big|   \lambda^{(\delta)}_n(\theta) - \delta^{-2} \mu^{(\delta)}_n(\theta) \big| \le C_1 \delta \big(   \lambda^{(\delta)}_n(\theta) + \delta^{-2} \mu^{(\delta)}_n(\theta) \big),
	\]
	for all $\t \in \square^* \backslash \{0\}$ and $\delta \in (0,\pi)$. 
\end{thm}
This, in particular, shows that
\begin{equation}\label{26.4.24f2}
  \lambda^{(\delta)}_n(\theta)  =\Big(1+  \mathcal{O}(\delta)\Big) \delta^{-2} \mu^{(\delta)}_n(\theta), 
\end{equation}
uniformly in $\t$ and $n$, as $\delta$ tends to zero.

Let us study the asymptotics of $\mu^{(\delta)}_n(\t)$ further. 	Notice that $\nu_{\t,\delta}$ is a, uniformly in $\t$, large parameter for small $\delta$. Indeed, from \eqref{19.3.24e2} (in the appendix), 
\[
g_\t(0) \ge (2\pi)^{-1} \ln{\pi} >0,
\]
and so 
\begin{equation}\label{26.4.24e1}
	\nu_{\t,\delta} \ge (2\pi)^{-1} \ln(\pi/\delta).
\end{equation}

As such we have a problem of the type studied, by the authors, in  \cite[Chapter 3]{SCIKVPS23} and we shall follow the spirit the arguments used therein; however, for this particular, example we shall prove slightly more.

First note that $\mathfrak{b} $ is an inner product on $L^2(\Omega;\CC^2)$; indeed, for $f \in C^\infty_0(\Omega,\CC^2)$ and $(p,q) = (Af_1,Af_2)$, and $\alpha\gg 1$, by repeated integration by parts and the behaviour of $p,q$ for large $r$, one has 
\[
\int_{B_\alpha} p_{,kl} \overline{q_{,nm}} = \int_{S_\alpha}p_{,kl} \overline{q_{,n}} -  \int_{S_\alpha}p_{,km} \overline{q_{,n}} + \int_{B_\alpha} p_{,km} \overline{q_{,nl}} = \mathcal{O}(\alpha^{-2}) + \int_{B_\alpha} p_{,km} \overline{q_{,nl}},
\]
for all $i,j,k,m \in \{1,2\}$, and therefore
\[
\mathfrak{b}[f] = \gamma^{-1} \int_\Omega |f|^2 + \int_\Omega  \big( |(Af_1)_{,12}+(Af_2)_{,11}|^2 + |(Af_1)_{,22}+(Af_2)_{,12}|^2 \big), \quad \forall f \in L^2(\Omega;\CC^2).
\]
Now, note that the space  $L^2_0(\Omega;\CC^2) : = \{ f \in L^2(\Omega;\CC^2) \, | \,  \int_\Omega f_1=\int_\Omega f_2= 0 \}$ is precisely where $\mathfrak{c}_1$ vanishes. Thus, we
 shall decompose $L^2(\Omega;\CC^2)$ into $L^2_0(\Omega;\CC^2)$ and its  orthogonal complement with respect to the inner product $\mathfrak{b}$. More precisely, let $N: \CC^2 \mapsto L^2_0(\Omega;\CC^2)$ be the linear operator with the action $c \mapsto Nc$ the unique solution to 
\[
\mathfrak{b}(Nc,\tilde{g})  = - \mathfrak{b}(c,\tilde{g}), \quad \forall \tilde{g} \in L^2_0(\Omega;\CC^2).
\]

By construction
\begin{equation*}
	L^2(\Omega;\CC^2) = L^2_0(\Omega;\CC^2) \dot{+} (I+N)\CC^2,
\end{equation*}
\begin{equation*}
	\mathfrak{b}( \tilde{g}, (I+N)\tilde{c}) = 0, \quad \forall \tilde{g}\in L^2_0(\Omega;\CC^2) , \quad \forall \tilde{c}\in \CC^2.
\end{equation*}
Furthermore,
\begin{equation*}
	\mathfrak{c}_1[(I+N)c +g] = \mathfrak{c}_1[c], \quad  \forall \tilde{c}\in \CC^2.
\end{equation*}

Whence,
\[
\mu^{(\delta)}_n(\t)  =    \min_{\substack{V \subset L^2_0(\Omega;\CC^2)\dot{+}\CC^2, \\ \dim{V} = n}} \max_{(g,c)\in V} \frac{ \mathfrak{b}[(I+N)c] + \mathfrak{b}[g]}{\nu_{\t,\delta}	\mathfrak{c}_1[c] +  \mathfrak{c}_2[(I+N)c+g] }.
\]
Since $\nu_{\t,\delta}$ is large, we can discard the cross terms in the denominator. Indeed, for
\begin{equation}\label{26.4.24f1}
\Lambda^{(\delta)}_n : = \min_{\substack{V \subset L^2_0(\Omega) \dot{+} \CC^2, \\ \dim{V} = n}} \max_{(g,c)\in V} \frac{ \mathfrak{b}[(I+N)c] +\mathfrak{b}[g]}{\nu_{\t,\delta}	\mathfrak{c}_1[c] +  \mathfrak{c}_2[g] },
\end{equation}
the following result holds:
\begin{prop}\label{7.6e.241}There exists a constant $C_2>0$ such that 
\begin{equation}\label{26.4.24f3}
	|\mu^{(\delta)}_n-\Lambda^{(\delta)}_n(\t) | \le C_2  \ln(\pi/\delta)^{-1/2}  \mu^{(\delta)}_n(\t),
\end{equation}
	for all $\delta \in (0,\pi/e)$ and all $\t \in \square^*\backslash\{0\}$.
\end{prop}
Note that \eqref{26.4.24f3} implies that
\begin{equation}\label{7.6.24e2}
\mu^{(\delta)}_n = \Big(1+ \mathcal{O}\left(\ln(\pi/\delta)^{-1/2}\right) \Big)\Lambda^{(\delta)}_n(\t)  \qquad \text{as $\delta \rightarrow 0$,}
\end{equation}
 uniformly in $\t$ and $n$.
\begin{proof}[Proof of Proposition \ref{7.6e.241}]
Let $\kappa_1 >0$ be the smallest  constant such that 
\begin{equation}\label{25.4.24e1}
	\mathfrak{c}_2[(I+N)c] + \mathfrak{b}[(I+N)c]\le \kappa_1^2 \mathfrak{c}_1[c], \quad \forall c \in \CC^2.
\end{equation}
By \eqref{25.4.24e1}, Cauchy-Schwarz inequality and then \eqref{26.4.24e1}, we compute
\begin{flalign*}
	\big| \mathfrak{c_2}[(I+N)c+g] - \mathfrak{c}_2[g] \big| & = \big|   \mathfrak{c}_2[(I+N)c]+2\mathrm{Re}\, \mathfrak{c}_2((I+N)c,g)  \big|\\
	&\le  \kappa_1^2  \mathfrak{c}_1[c] +2\kappa_1\sqrt{\mathfrak{c}_1[c]}\sqrt{\mathfrak{c}_2[g]}\\
	&\le  \kappa_1^2  \mathfrak{c}_1[c] +\kappa_1\nu^{-1/2}_{\t,\delta} \big( \nu_{\t,\delta}\mathfrak{c}_1[c]+ \mathfrak{c}_2[g]  \big) \\
	& \le\kappa_2 \ln(\pi/\delta)^{-1/2}\big(  \nu_{\t,\delta}\mathfrak{c}_1[c]+\mathfrak{c}_2[g] \big),
\end{flalign*}
for $\kappa_2 = \sqrt(2\pi)(\sqrt(2\pi)\kappa_1^2+\kappa_1) $ when $\delta < \pi/e$.

Therefore
\begin{flalign*}
	\left| \frac{\mathfrak{b}[(I+N)c]+ \mathfrak{b}[g]}{\nu_{\t,\delta}	\mathfrak{c}_1[c] +  \mathfrak{c}_2[(I+N)c+g] } - \frac{\mathfrak{b}[(I+N)c]+ \mathfrak{b}[g]}{\nu_{\t,\delta}	\mathfrak{c}_1[c] +  \mathfrak{c}_2[g]   } \right| 
	& \le \kappa_2 \ln(\pi/\delta)^{-1/2}\left(\frac{ \mathfrak{b}[(I+N)c]+\mathfrak{b}[g]}{\nu_{\t,\delta}	\mathfrak{c}_1[c] +  \mathfrak{c}_2[(I+N)c+g] }  \right),
\end{flalign*}
and  the desired result holds.
\end{proof}

Now, as the right-hand-side of \eqref{26.4.24f1} is diagonalised, $\Lambda^{(\delta)}_n(\t)$, $n \in \NN$, are the critical points of 
\[
\mathcal{I}_1(c) : = \frac{ \mathfrak{b}[(I+N)c]}{\nu_{\t,\delta}	\mathfrak{c}_1[c]},
\]
and 
\[
\mathcal{I}_2(g):= \frac{\mathfrak{b}[g]}{ \mathfrak{c}_2[g] }.
\]
Moreover, by elliptic regularity, there exists a constant $\kappa_3>0$ such that
\[
\kappa_3 \mathfrak{c}_2[f] \le \mathfrak{b}[f], \quad \forall f\in L^2(\Omega;\CC^2),
\]
and so
\[
\inf_{g\in L^2_0(\Omega;\CC^2)} \frac{\mathfrak{b}[g]}{\mathfrak{c}_2[g]}\ge \kappa_3.
\]
Additionally, by \eqref{26.4.24e1} and \eqref{25.4.24e1},
\begin{equation}\label{26.4.24e5}
	\max_{c\in\CC^2} \frac{\mathfrak{b}[(I+N)c]}{\nu_{\t,\delta}\mathfrak{c}_1[c]} \le \frac{2\pi\kappa_1^2}{\ln(\pi/\delta)}.
\end{equation}
Consequently, the critical points of $\mathcal{I}_1$ are below the critical points of $\mathcal{I}_2$ for all $\delta <\pi \exp(-2\pi \kappa_1^2/\kappa_3)$. That is, for such $\delta$, one has 
\begin{align}
	\Lambda^{(\delta)}_n(\t) &= \nu_{\t,\delta}^{-1} \Lambda_n, \quad n=1,2; \label{26.4.24e3}\\
	\Lambda^{(\delta)}_n(\t) & = \Lambda_{n} := \min_{\substack{V \subset L^2_0(\Omega;\CC^2), \\ \dim{V} = (n-2)}} \max_{g\in V} \frac{ \mathfrak{b}[g]}{  \mathfrak{c}_2[g] }, \quad n= 3,4,\ldots. \label{26.4.24e4} 
\end{align}
where $0\le \Lambda_1 \le\Lambda_2$ are the  critical points of 
\[
\frac{ \mathfrak{b}[(I+N)c]}{	\mathfrak{c}_1[c]  },
\]
i.e., the roots of 
\[
p(\mu) := \det \Big| \mu |\Omega|^2 I -B \Big|,
\]
for the symmetric matrix
\[
B : = \left( \begin{matrix}
	\mathfrak{b}[(I+N)e_1] & \mathfrak{b}\big((I+N)e_1, (I+N)e_2 \big) \\
	\mathfrak{b}\big((I+N)e_2, (I+N)e_1 \big) & \mathfrak{b}[(I+N)e_2]
\end{matrix} \right),
\]
where $e_1,e_2$ are the standard Euclidean basis vectors. 

In particular, \eqref{26.4.24f2}, \eqref{7.6.24e2}, \eqref{26.4.24e5}, \eqref{26.4.24e3} and \eqref{26.4.24e4}   readily imply the following result\footnote{By $h=\mathcal{O}_n(\delta)$ we mean that there exists a positive constant that depends on $n$, $C(n)$, such that $\frac{|h|}{\delta}\le C(n)$ for sufficiently small $\delta$. }:
\begin{thm}One has the following, uniform in $\t$, asymptotic behaviour for small $\delta$:
\[
\lambda^{(\delta)}_{n}(\t) =\frac{1}{\delta^{2}\big( g_\t(0) - (2\pi)^{-1} \ln{\delta}\big)} \Lambda_n + \mathcal{O}(\delta^{-2}|\ln{\delta}|^{-3/2}), \quad n=1,2,
\]
and
\[
\lambda^{(\delta)}_n(\t) =\delta^{-2} \Lambda_n + \mathcal{O}_n(\delta^{-2} |\ln{\delta}|^{-1/2}), \quad n \ge 3,
\]
as $\delta \rightarrow 0$.
\end{thm}
This result readily implies  the existence of a (low frequency) gap in $\sigma_0$ for sufficiently small $\delta$:
\begin{cor}
There exists a $\delta_0>0$ such that
\[
\lambda^{(\delta)}_2(\t) <\lambda^{(\delta)}_3(\t) 
\]
for all $\delta<\delta_0$ and all $\t \in \square^*$.
\end{cor}

\begin{rem}
When $\Omega = B_1$, the ball of radius 1 centred at the origin, one has
\[
A1 = \left\{ \begin{array}{lcl}
	-	\frac{|x|^2}{4} + \frac{1}{4}  & & \text{in $\Omega$} \\[0.2em]
	-\frac{1}{2}\ln{|x|}&  & \text{in $\RR^2 \backslash \Omega$}
\end{array} \right. .
\]
Then, direct calculation gives
\[
\mathfrak{b}(\tilde{c},\tilde{g}) = 0, \quad \forall \tilde{c} \in \CC^2, \forall \tilde{g} \in L^2_0(\Omega;\CC^2).
\]
Therefore, $N \equiv 0$ and, by  direct calculation, one has $B =  \big(\gamma^{-1}+\frac{1}{2}\big)I$ and consequently 
\[
\Lambda_1  = \Lambda_2 =  \left(\gamma^{-1}+\frac{1}{2}\right).
\]
\end{rem}

\section{Appendix}

\subsection{Equivalence of \eqref{degp} and the reduced Maxwell system \eqref{Max1}-\eqref{Max2} on the critical light line.}
Here we establish the following result:
\begin{lem}\label{BthetaIsmaxwell}
Let $k=1$ and $\omega$ be such that $(\omega^2,v)$ is an eigenvalue-eigenfunction pair of $\B(\t)$ for some $\t \in \square^*$.  Then,  there exist solutions of the form \eqref{propsol}, \eqref{qpbc} to the Maxwell system \eqref{coef}, \eqref{Maxwellsystem}.
\end{lem}
To prove this result, we shall introduce $W(\t)$, the orthogonal complement of $V(\t)$ in $H^1_\t(\square;\CC^2)$ with respect to the equivalent inner product
\begin{equation}\label{Hnorm}
(u,\tilde{u})_{H^1_\t(\square;\CC^2)} = \left\{ \begin{array}{lcc}
\int_\square \nabla u: \overline{\nabla \tilde{u}} & & \t \neq 0, \\[.2em]
\int_\square \nabla u: \overline{\nabla \tilde{u}} + \left| \int_\square u \right|\left| \int_\square \tilde{u} \right|& & \t = 0, 
\end{array} \right.
\end{equation}
and utilise the following fact:
\begin{prop}\label{Wrep}
For any $\t \in \square^*$, $w \in W(\t)$ if, and only if,
\[
\text{$\div{w}=0\quad $  and  $\quad \curl{w}=0\quad $ in $C$},
\] 
with additionally $\int_\square w =0$ when $\t =0$. 
\end{prop}
\begin{proof}
For any $\t \in \square^*$, $w\in W(\t)$ and $v \in V(\t)$ one has (cf. \eqref{ibp})
\[
0 =  \int_R \div{w} \, \overline{f_1} + \curl{w} \, \overline{f_2} + 2\pi \delta_{\theta0} \Big(\int_\square w\Big) c ,
\]
for arbitrary $f_1, f_2 \in L^2(R)$, $c \in \CC^2$ and $\delta$ is the Kronecker delta. Here, we utilised the representation for $V(\t)$ given by Lemma \ref{lem:spcom2}  The desired results readily follow. 
\end{proof}
\begin{proof}[Proof of Lemma \ref{BthetaIsmaxwell}]
Set $H_3 = v_1,$ $E_3 = v_2$. Then, the systems in \eqref{Max1} are solvable and a general form of the solutions (see Section \ref{s:k1}) is
\begin{equation}\label{HEforms}
\begin{aligned}
	\left( \begin{array}{c}
		H_1 \\  E_2
	\end{array}\right) &=  \chi_R \frac{\i}{\omega}\frac{1}{\gamma}\left( \begin{array}{c}
		H_{3,1} +\epsilon_0E_{3,2} \\ H_{3,1}+E_{3,2}
	\end{array}\right) + (1-\chi_R) \frac{\i}{2\omega} v_{2,2} \left( \begin{array}{c}
		1 \\  -1
	\end{array}\right) +\frac{\i}{\omega}c_1\left( \begin{array}{c}
		1 \\  1
	\end{array}\right) ,\\
	\left( \begin{array}{c}
		H_2 \\  E_1
	\end{array}\right) & = \chi_R \frac{\i}{\omega}\frac{1}{\gamma}\left( \begin{array}{c}
		H_{3,2}-\epsilon_0E_{3,1} \\ E_{3,1}-H_{3,2}
	\end{array}\right) - (1-\chi_R)\frac{\i}{2\omega}   v_{2,1} \left( \begin{array}{c}
		1 \\  1
	\end{array}\right) + \frac{\i}{\omega}c_2\left( \begin{array}{c}
		-1 \\  1
	\end{array}\right) ,
\end{aligned}
\end{equation}
for arbitrary functions $c_1, c_2$ supported in $\overline{\square \backslash R}$. 

	It remains to show that \eqref{Max2} holds, i.e.  
\begin{equation}\label{17dec3}
	\begin{aligned}
		\int_{\square}  \left( \begin{array}{c}
			H_1 \\  E_2
		\end{array}\right)  \cdot \overline{ \left( \begin{array}{c}
				\phi_{2,2} \\  \phi_{1,1}
			\end{array}\right) } -   \left( \begin{array}{c}
			H_2 \\  E_1
		\end{array}\right)  \cdot \overline{ \left( \begin{array}{c}
				\phi_{2,1} \\  \phi_{1,2}
			\end{array}\right) }  =  \i \omega\Big( \int_\square  v \cdot \overline{\phi}+ &\gamma \int_R  v_2 \overline{\phi_2} \Big) \\
	&	\quad \forall \phi \in H^1_{\t}(\square;\CC^2),
	\end{aligned}
\end{equation}
holds, for some particular choice of $c_1, c_2$.

 Substituting \eqref{HEforms} into \eqref{17dec3} gives
\begin{equation*}
\begin{aligned}
&\int_{\square \backslash R} c_1   \overline{ \div{\phi} } +c_2   \overline{ \curl{\phi} } +  \int_{\square\backslash R} v_{2,1} \overline{\left( \frac{\phi_{2,1}+\phi_{1,2}}{2} \right)}+ v_{2,2} \overline{\left( \frac{\phi_{2,2}-\phi_{1,1}}{2} \right) }  \\
&+ \int_R \nabla v_2 \cdot \overline{\nabla \phi_2}+ \gamma^{-1}\int_{R} \div v \, \overline{\div{\phi}} +\curl v \, \overline{\curl{\phi}} = \omega^2 \left( \int_\square  v \cdot \phi  + \gamma \int_R  v_2 \overline{\phi_2} \right)
\end{aligned}
\end{equation*}
$\forall \phi \in H^1_{\t}(\square;\CC^2)$. Decomposing $\phi$  with respect to the decomposition $H^1_{\t}(\square;\CC^2) = V(\t) \oplus W(\t)$,  utilising the fact $v$ is an eigenfunction of $\B(\t)$ with eigenvalue $\lambda = \omega^2$ (cf. \eqref{opdegp} and in particular \eqref{degp}) and utilising Proposition \ref{Wrep}, it follows that:   \eqref{17dec3} holds if, and only if, the pair$(c_1,c_2)$ satisfies  
\begin{equation}\label{19dec3}
	\begin{aligned}
\int_{\square \backslash R} c_1 &  \overline{ \div{\tilde{w}}} +c_2   \overline{ \curl{\tilde{w}}} = \langle F, \tilde{w} \rangle, \quad \forall \tilde{w} \in W(\t),
	\end{aligned}
\end{equation}
for
\[
\begin{aligned}
\langle F, \tilde{w} \rangle : = \omega^2  \Big( \int_{\square \backslash R}  v \cdot \overline{\tilde{w}} &+ \gamma \int_R v_2 \overline{\tilde{w}_2} \Big)   \\ &- \int_{\square\backslash R} v_{2,1} \overline{\left( \frac{\tilde{w}_{2,1}+\tilde{w}_{1,2}}{2} \right)}+ v_{2,2} \overline{\left( \frac{\tilde{w}_{2,2}-\tilde{w}_{1,1}}{2} \right) }- \int_R \nabla v_2 \cdot \overline{\nabla \tilde{w}_2}.
\end{aligned}
\]
One can now consider  $c_1= \div w$ and $c_2 =  \curl{w}$ for some $w \in W(\t)$ to be determined. Indeed, by Proposition \ref{Wrep}, $\div{w} = \curl{w}=0$ in $R$.  From \eqref{19dec3}, it follows that $w$ must solve 
\begin{equation}\label{19dec5}
	\begin{aligned}
		\int_{\square \backslash R}  \div{w}\, \overline {\div \tilde{w}} + \curl{w} \, \overline{\curl{w}}
	 = \langle F, \tilde{w} \rangle\quad \forall \tilde{w} \in W(\t).
	\end{aligned}
\end{equation}
Finally note that there exists a unique solution in $W(\t)$ to \eqref{19dec5}; indeed, the right-hand-side is a bounded linear functional on $W(\t)$ and the left-hand-side is simply the inner product $(w,\tilde{w})_{H^1_{\t}(\square;\CC^2)}$ (see \eqref{Hnorm} and \eqref{ibp}). Hence, \eqref{17dec3} holds when  $c_1 = \div{w}$, $c_2=\curl{w}$ and the proof is complete. 
	\end{proof}
\subsection{Proof of Lemma \ref{prethm.resolventasymp}.}\label{app2}
Here, we shall prove Lemma \ref{prethm.resolventasymp}.  For this, we shall rely on the techniques introduced in  \cite[Section 3]{SCIKVPS23}.

For fixed $0<\ep<1$, $\t \in \square^*$ and $F \in H^{-1}(\square;\CC^2)$, recall $U_{\ep,\t}$ is the solution to
\begin{equation}\label{15.6.24e4}
\ep^{-1} a(U_{\ep,\t},\phi) + b_{\ep}(U_{\ep,\t},\phi) + c(U_{\ep,\t},\phi) = \langle F,\phi\rangle, \quad \forall \phi \in H^1_\t(\square;\CC^2),
\end{equation}
and  $V_\t \in V(\t)$ is the solution to
\[
 b(V_\t,\tilde{V}) = \langle F,\tilde{V} \rangle, \quad \forall\, \tilde{V} \in V(\t).
\]
We shall prove the desired estimate by considering  an intermediary function, $W_{\ep,\t} \in H^1_\t(\square;\CC^2)$,  the unique solution  
\begin{equation}\label{15.6.24e6}
	\ep^{-1} a(W_{\ep,\t},\phi) + b(W_{\ep,\t},\phi) + c(W_{\ep,\t},\phi) = \langle F,\phi\rangle, \quad \forall \phi \in H^1_\t(\square;\CC^2),
\end{equation}
Note that, $W_{\ep,\t}$ exists because $a+b+c$ is an equivalent inner product on $H^1(\square;\CC^2)$. Indeed, by arguing as in the proof of \eqref{epformlowerbound}  it follows that 
\begin{equation}\label{16.5.24e10}
	\nu_1\| \nabla \phi \|^2_{L^2(\square;\CC^2)} \le a[\phi]+b[\phi] \le \nu_2 \| \nabla \phi \|^2_{L^2(\square;\CC^2)}, \quad \forall \phi \in H^1(\square;\CC^2),
\end{equation}
for $\nu_1 =\min\{ \tfrac{1}{2}, (1-\eta)/\gamma\}$ and $\nu_2 =  \max\{3/2,1 +2/\gamma \} $. 

We begin the proof by  noting that $b$ (cf. \eqref{prebform}) is the continuous extension of $b_\ep$ (cf. \eqref{bepform}) to $\ep=0$; indeed, we demonstrate below the following result:
\begin{prop}\label{relboundb}
	There exists a $L_b>0$ such that 
	\begin{equation}\label{rbineq}
		\left| b_\ep(\phi,\psi) - b(\phi,\psi) \right| \le   \ep L_b\sqrt{\ep^{-1}a[\phi]+b_\ep[\phi]}  \sqrt{a[\psi]+b[\psi]}\quad \forall \phi,\psi \in H^1(\square;\CC^2),
	\end{equation}
	and for all $0<\ep<1$.
\end{prop}
From this `relative' bound of forms,  we can readily approximate $U_{\ep,\t}$ by $W_{\ep,\t}$.
Indeed, by \eqref{15.6.24e4}, \eqref{15.6.24e6} and \eqref{rbineq},  $r : = U_{\ep,\t} -W_{\ep,\t}$ satisfies
\begin{flalign*}
	\ep^{-1}a[r] + b[r]+c[r] = b(U_{\ep,\t},r)- b_{\ep}(U_{\ep,\t},r)   \le \ep L_b \sqrt{\ep^{-1}a[U_{\ep,\t}]+b_\ep[U_{\ep,\t}]} \sqrt{a[r]+b[r]}
\end{flalign*}
which, along with the standard inequality \[\ep^{-1}a[U_{\ep,\t}]+b_\ep[U_{\ep,\t}] + c[U_{\ep,\t}]\le C_1^2 \| F\|^2_{H^{-1}(\square;\CC^2)}\] for $C_1^2 = \max\{1,\nu^{-1}\}$ (cf. \eqref{15.6.24e4}, \eqref{cform} and \eqref{epformlowerbound}), gives 
\[
 a[r] + b[r]+c[r]  \le   \ep^2 L_b^2 C_1^2  \| F\|^2_{H^{-1}(\square;\CC^2)}.
\]
This,  along with \eqref{cform} and\eqref{16.5.24e10}, implies that $r = U_{\ep,\t}-W_{\ep,\t}$ satisfies
\begin{equation*}
\|U_{\ep,\t}-W_{\ep,\t}\|_{H^1(\square;\CC^2)}  \le   \ep \frac{L_b C_1}{C_2} \|F \|_{H^{-1}(\square;\CC^2)},
\end{equation*}
for $C_2^2 = \min\{1,\nu_1\}$.

To complete the proof of Lemma \ref{prethm.resolventasymp}, it remains to show that there exists a $\nu_\star >0$ independent of $\ep$, $\t$ and $f$ such that
\begin{equation}\label{15.6.24e8}
 a[W_{\ep,\t}-V_{\t}] + b[W_{\ep,\t}-V_{\t}] + c[W_{\ep,\t}-V_{\t}] \le   \ep^2 \nu_\star^{-2} \|F \|^2_{H^{-1}(\square;\CC^2)}.
\end{equation}
for  some $\nu_*$ independent of $\ep, \t$ and $F$.  To do this, we shall utilise \cite[Theorem 3.1]{SCIKVPS23}; setting,
\[
a_\t(\phi,\psi) : = a(e^{\i \t \cdot y}\phi ,e^{\i \t \cdot y}\psi), \qquad  \text{and} \quad \ b_{\t}(\phi,\psi) : = b(e^{\i \t \cdot y}\phi ,e^{\i \t \cdot y}\psi) + c(\phi, \psi),
\]
then, we note that $W_{\ep,\t} = e^{\i \t \cdot y} u_{\ep,\t}$ for $u_{\ep,\t} \in H^1_{per}(\square;\CC^2)$ the solution to 
\begin{equation*}
\ep^{-1} a_\t(u_{\ep,\t}, \phi) + b_{\t}(u_{\ep,\t}, \phi)   = \langle f, \phi\rangle, \quad \forall \phi \in H^1_{per}(\square;\CC^2);
\end{equation*}
and that  $V_\t = e^{\i \t \cdot y} v_\t$ where $v_\t \in \mathcal{V}(\t)$ is the solution to
\[
b_\t(v_\t, \tilde{v})  = \langle f, \tilde{v}\rangle, \quad \forall \tilde{v} \in \mathcal{V}(\t).
\]
Here, 
\[
\mathcal{V}(\t) : = \{ v\in H^1_{per}(\square;\CC^2)\, | \, \div{v} + \i \t \cdot v = 0, \text{ and } \curl{v} + \i \t^\perp \cdot v =0 \text{ in $\square \backslash R$} \}.
\]
We now  argue as in \cite[Section 3]{SCIKVPS23}. First note that the closed subspace  $\mathcal{V}(\t)$ is the kernel of the form $a_\t$, precisely
\[
\mathcal{V}(\t) ;= \{ v\in H^1_{per}(\square;\CC^2) \, : \, a_\t[v]=0 \}.
\]
Furthermore, it can be seen that $a_\t$ is coercive on $(\mathcal{V}(\t))^\perp$, the orthogonal complement of $\mathcal{V}(\t)$ in $H^1_{per}(\square;\CC^2)$, with coercivity constant $\nu(\t)>0$. In fact, we can say much more; the form $a_\t$ is \textit{uniformly} coercive on $(\mathcal{V}(\t))^\perp$, i.e. $\nu_* := \inf_{\t} \nu(\t)>0$. Indeed, it was shown in \cite[Theorem 3.4]{SCIKVPS23} that the uniform coercivity of $a_\t$ on the orthogonal complement to its kernel is equivalent to the kernel being continuous with respect to $\t$. We demonstrated in Theorem \ref{lem:spcom1} that $V(\t)$ was Lipschitz continuous in $\t$, thus clearly $\mathcal{V}(\t)$ is Lipschitz in $\t$ and consequently Theorem 3.2 of \cite{SCIKVPS23} holds, which, in this setting, implies:
\begin{equation*}
a_\t[u_{\ep,\t} - v_{\ep,\t}] +b_\t[u_{\ep,\t} - v_{\ep,\t}] \le \nu_\star^{-2} \ep^2 \| F\|_{H^{-1}(\square;\CC^2)}.   
\end{equation*}
That is \eqref{15.6.24e8} holds and the proof of Lemma \ref{thm.resolventasymp} is complete. 

\begin{proof}[Proof of Proposition \ref{relboundb}]
	By \eqref{epformlowerbound} and \eqref{16.5.24e10}, it is sufficient to prove 
	\begin{equation}\label{15.6.24.e5}
		\left| b_\ep(\phi,\psi) - b(\phi,\psi) \right| \le \ep  L_0  \| \nabla \phi \|_{L^2(\square;\CC^2)} \| \nabla \psi \|_{L^2(\square;\CC^2)}\quad \forall \phi,\psi \in H^1(\square;\CC^2),
	\end{equation}
	for some $L_0>0$. Now, from the inequalities
	\[
|	\sqrt{1-\ep} - 1| \le \ep, \quad \text{and} \quad |	\sqrt{1-\ep} - (1-\tfrac{\ep}{2})| \le \frac	{\ep^2}{2}, \quad \ep \in (0,1), 
	\]
we find (cf. \eqref{prebform} and \eqref{bepform})  that
\begin{flalign*}
&	\left| b_\ep(\phi,\psi) - b(\phi,\psi) \right|   \le  \frac{\ep}{2}\int_{\square \backslash R} | \nabla \phi_1^\perp \cdot \overline{\nabla \psi_2}+\nabla \phi_2 \cdot \overline{\nabla^\perp \psi_1}| + \\
&+\frac{\ep}{\gamma(\gamma+\ep)}\int_R |\nabla \phi_1 \cdot \overline{\nabla \psi_1} +\epsilon_0\nabla \phi_2 \cdot \overline{\nabla \psi_2} | + \frac{(\gamma+1)\ep}{\gamma(\gamma+\ep)}\int_{R} | \nabla \phi_1^\perp \cdot \overline{\nabla \psi_2}+\nabla \phi_2 \cdot \overline{\nabla^\perp \psi_1}| \\
& \le 
 \frac{\ep}{2}\int_{\square \backslash R} | \nabla \phi| |\nabla \psi| +\frac{2\epsilon_0\ep}{\gamma^2}\int_R |\nabla \phi||\nabla \psi|,
\end{flalign*}
where the last inequality followed from $ \epsilon_0=\gamma+1>1$ and the Cauchy-Schwarz inequality. Finally, by an application of H\"{o}lder's inequality,  \eqref{15.6.24.e5} holds for $L_0 =\max\{ \frac{1}{2} ,2\epsilon_0\gamma^{-2} \}$. 
\end{proof}

\subsection{Proof of Theorem \ref{thmArrowAsymptotics}} \label{App:Arrow}
 Throughout,  for a set $\mathcal{O}$,   $\langle \cdot, \cdot \rangle_\mathcal{O}$ and $\| \cdot \|_\mathcal{O}$ will denote the standard $L^2(\mathcal{O})$ inner product and norm respectively;  $\chi_{\mathcal{O}}$ shall denote the operation that extends a function by zero outside of $\mathcal{O}$. $B_\alpha$ and $S_\alpha$ will denote the ball and circle, respectively, of radius $\alpha$ centred at the origin. \\

 Fix $\theta \neq 0$, $0<\delta<\pi$. For $i=1,2$, consider $f_i \in L^2(R_\delta)$ where $R_\delta: = \delta \Omega$ and $\Omega \subset B_1$. Recall $ A_\t f_i = u_i \in H^1_\theta(\square)$,  are the solutions to
\[
-\Delta u_i = \chi_{R_\delta}f_i \quad \text{in $\square$}.
\]
Let us introduce the solutions $v_i$ to 
\begin{equation*}
	\begin{aligned}
-\Delta v_i = \chi_{R_\delta} f_i \quad \text{in $\RR^2$ with $v_i(x) = - \frac{d_i}{2\pi}\ln{|x|} + o(1)$ as $|x| \rightarrow + \infty$.}
\end{aligned}
\end{equation*}
Clearly, $d_i = \langle  f_i, 1 \rangle_{R_\delta}$.  

Furthermore, we shall denote $v_\i = A_\delta f_i$. Notice, $A_\delta$ is non-negative and self-adjoint.

Let us introduce the non-negative sesquilinear forms

\begin{equation}\label{20.6.24e3}
	\begin{aligned}
	{b}^{(\delta)}(f, \tilde{f}) : = \gamma^{-1} \int_{R_\delta}  f \cdot \overline{ \tilde{f}}- \int_{R_\delta}\left[ \Big((A_\delta f_1)_{,22} + (A_\delta f_2)_{,12}\Big) \overline{\tilde{f}_1} + \Big((A_\delta f_1)_{,21} + (A_\delta f_2)_{,11}\Big) \overline{\tilde{f}_2} \right],
	\end{aligned}
\end{equation}
and
\begin{multline}\label{cdform}
c_{\delta,\theta}(f,\tilde{f}) := g_\t(0)  \Big(  \langle f_1 ,1\rangle_{R_\delta}\overline{\langle \tilde{f}_1,1 \rangle_{R_\delta}}   + \langle f_2 ,1\rangle_{R_\delta}\overline{\langle \tilde{f}_2,1 \rangle_{R_\delta}} \ \Big)\\+ \int_{R_\delta} \Big( f _1 \overline{A_\delta\tilde{f_1}} +  f_2 \overline{A_\delta\tilde{f}_2} \Big) + \gamma \int_{R_\delta} \Big((A_\delta f_1)_{,2} +(A_\delta f_2)_{,1}\Big) \overline{\Big((A_\delta \tilde{f}_1)_{,2} +(A_\delta \tilde{f}_2)_{,1}\Big)}.
\end{multline}
where we recall $g_\t$, the regular part of the quasi-periodic Green's function of the Laplacian (cf. \eqref{regG}). We shall see below, precisely \eqref{19.3.24e2}, that
\begin{equation}\label{gbound}
g_\t(0) \ge \frac{1}{2\pi}\ln{\pi};
\end{equation}

To prove Theorem \ref{thmArrowAsymptotics}, it is sufficient to prove the following result:

\begin{lem}\label{lem.prob1}
\[
\left| \frac{b_\t[f]}{c_\t[f]} - \frac{b^{(\delta)}[f]}{c_{\delta,\theta}[f]}   \right| \le C_1 \delta \left( \frac{b_\t[f]}{c_\t[f]}  +\frac{b^{(\delta)}[f]}{c_{\delta,\theta}[f]}\right), \qquad \forall f \in L^2(R_\delta;\CC^2),
\]
for some $C_1< 1+3 \gamma$.
\end{lem}

Indeed, Lemma \ref{lem.prob1} implies 
	\[
| \lambda^{(\delta)}_n(\theta)  - \nu^{(\delta)}_n(\theta) | \le C_1\delta \Big( \lambda^{(\delta)}_n(\theta)  + \Lambda^{(\delta)}_n(\theta) \Big),
\]
for
\[
\nu^{(\delta)}_n(\theta) : =  \min_{\substack{V \subset L^2(\Omega;\CC^2), \\ \dim{V} = n}}\frac{b^{(\delta)}[f]}{c_{\delta,\theta}[f]}.
\]

Furthermore, as $R_\delta = \delta \Omega$, then by the rescaling $x = \delta y$ for $y \in \Omega$, we find
\[
\Lambda^{(\delta)}_n(\t)  =  \delta^{-2} \mu^{(\delta)}_n(\t),
\]
where $\mu^{(\delta)}_n(\t)$ is given in \eqref{minmax2}. That is, Theorem \ref{thmArrowAsymptotics} holds.

For ease of exposition we shall use, henceforth, $\langle \cdot, \cdot \rangle$ to denote  $\langle \cdot, \cdot \rangle_{R_\delta}$. Let us prove Lemma \ref{lem.prob1}. The following, key estimate, is instrumental to the proof: 
\begin{lem}\label{mainestimate} For any $ \alpha \in(\delta,\pi)$, one has
\[
\big| \langle f_i , u_j \rangle - \langle f_i, v_j \rangle - g_\t(0) d_i\overline{d_j} \big| \le 3 \frac{\delta}{\alpha} \sqrt{ \langle f_i, v_i \rangle + g_\t(0) |d_i|^2} \sqrt{\langle f_j,v_j \rangle + g_\t(0) |d_j|^2}.
\]
\end{lem}
\begin{rem}
Notice $ \langle f_i, v_i \rangle + g_\t(0) |d_i|^2 \ge 0$ since $A_\delta$ is non-negative and $g_\t(0)$ is positive (see \eqref{gbound}).
\end{rem}
To prove this key estimate we shall use the following inequalities which directly follow from basic properties of harmonic functions:
\begin{prop}
Let $v_i$ and $g_\t$ be as above, $ \alpha \in [\delta, \pi)$ and  $h$ an harmonic function in $B_\alpha$. Then

\begin{align}
\label{harm1}
\langle \nabla h, \nabla h \rangle_{B_\delta}  & \le \left( \frac{\delta}{\alpha} \right)^2\langle \nabla h, \nabla h \rangle_{B_\alpha}; \\
\label{harm2}
\langle \partial_r v_i, h \rangle_{S_\beta} & = -  	\langle v_i, \partial_r h \rangle_{S_\beta} - d_i \overline{h(0)}; \quad \text{for $\delta \le \beta \le \alpha$;} \\
	\| \nabla v_i \|_{L^2(B_\alpha)} ^2 &\le \langle f_i,  v_i \rangle+ |d_i|^2 \frac{1}{2\pi} \ln{\alpha}; \label{19.3.24e1}\\
	\| \nabla g_\t \|_{L^2(B_\alpha)} ^2 &\le   g_\t(0)  - \frac{1}{2\pi} \ln{\alpha}. \label{19.3.24e2} 	
\end{align}
\end{prop}
\begin{proof}
We  have the following polar representations:
\begin{equation}\label{15.4.24e1}
v_i(r,\theta) = - \frac{d_i}{2\pi} \ln{r} + \sum_{n=1}^\infty c_n r^{-n} e^{\mathrm{i}  n \theta}, \qquad \text{in $\RR^2 \backslash B_\delta$},
\end{equation}
and 
 \begin{equation}\label{15.4.24e2}
h(r,\t) =h(0) +  \sum_{n=1}^\infty h_n r^{n} e^{\mathrm{i}  n \theta}, \qquad \text{in $ B_\alpha$},
 \end{equation}
 for some constants $c_n, h_n$. 
 
 \textit{Proof of \eqref{harm1}.} This clearly follows from \eqref{15.4.24e2}. Indeed, direct calculation
 \[
 \langle \nabla h, \nabla h \rangle_{B_\beta} = \langle \partial_r h, h \rangle_{S_\beta} =   2\pi\sum_{n=1}^\infty n \beta^{2n} |h_n|^2 \qquad \text{for $\delta \le \beta \le \alpha$.}
 \]
 \textit{Proof of \eqref{harm2}.} Since $v$ and $h$ are both harmonic in the annulus $B_\alpha \backslash B_\beta$, Green's second identity gives 
 \[
 \langle \partial_r v_i, h \rangle_{S_\beta}  -  	\langle v_i, \partial_r h \rangle_{S_\beta} =  \langle \partial_r v_i, h \rangle_{S_\alpha}  -  	\langle v_i, \partial_r h \rangle_{S_\alpha}.
 \]
 Now, by \eqref{15.4.24e1} and \eqref{15.4.24e2}, direct calculation gives 
 \[
 \langle \partial_r v_i, h \rangle_{S_\alpha} = - d_i\overline{h(0)} -  2\pi \sum_{n=1}^\infty n c_n \overline{h_n} 
 \]
 and 
  \[
 \langle  v_i, \partial_r  h \rangle_{S_\alpha} = 2\pi \sum_{n=1}^\infty n c_n \overline{h_n}.
 \]
 The above assertions imply \eqref{harm2}.

 \textit{Proof of \eqref{19.3.24e1}.} From \eqref{15.4.24e1}, one computes 
 \[
 \langle  \partial_r v_i, v_i \rangle_{S_\alpha} = |d_i|^2 \frac{\ln{\alpha}}{2\pi} - 2\pi \sum_{n=1}^\infty n |c_n|^2 \alpha^{-2n}  \le  |d_i|^2 \frac{\ln{\alpha}}{2\pi}
 \]
 Then
 \[
 \langle f_i,v_i\rangle = \langle \chi_{R_\delta}f_i,v_i\rangle_{B_\alpha} = \langle -\Delta v_i,v_i\rangle_{B_\alpha}  = \langle \nabla v_i, \nabla v_i \rangle_{B_\alpha} - \langle \partial_r v_i, v_i\rangle_{S_\alpha},  
 \]
 and \eqref{19.3.24e1} follows. 
 
 \textit{Proof of \eqref{19.3.24e2}.} By the quasi-periodicity of $G_\t$,  
 \begin{flalign*}
 	0 & = \langle  G_\t, - \Delta G_\t \rangle_{\square \backslash B_\alpha} = \| \nabla G_\t \|_{\square \backslash B_\alpha}^2 + \langle  G_\t,\partial_r G_\t \rangle_{S_\alpha},
 \end{flalign*} 
 and, since $G_\t = - 2\pi^{-1} \ln{r} + g_\t$,  
 \begin{flalign*}
 	\langle G_\t,\partial_r G_\t\rangle_{S_\alpha} &=  \langle   2\pi^{-1} \ln{r},\partial_r 2\pi^{-1} \ln{r} \rangle_{S_\alpha} - \langle   2\pi^{-1} \ln{r}, \partial_r g_\t \rangle_{S_\alpha}-  \langle   g_\t, \partial_r 2\pi^{-1} \ln{r} \rangle_{S_\alpha}+  \langle  g_\t, \partial_r g_\t \rangle_{S_\alpha}.
 \end{flalign*}
 Clearly,
 \[
 \langle   2\pi^{-1} \ln{r}, \partial_r2\pi^{-1} \ln{r} \rangle_{S_\alpha} = 2\pi^{-1} \ln{\alpha},
 \] and, since $g_\t$ is harmonic in $\square$, then 
 \[
 \langle g_\t, \partial_r g_\t \rangle_{S_\alpha} = \| \nabla g_\t\|_{B_\alpha}^2.
 \]
 Furthermore, using representation \eqref{15.4.24e2} (for $h=g_\t$) gives 
 \[
 \langle g_\t,  \partial_r 2\pi^{-1} \ln{r} \rangle_{S_\alpha}=g_\t(0) \quad \text{and} \quad \langle  2\pi^{-1} \ln{r} ,\partial_r g_\t \rangle_{S_\alpha}=0.
 \]
 Therefore, putting the above together gives
 \[
 0 = \| \nabla G_\t \|_{\square \backslash B_\alpha}^2 +  2\pi^{-1} \ln{\alpha} -   g_\t(0) +  \| \nabla g_\t\|_{B_\alpha}^2;
 \]
and   \eqref{19.3.24e2} follows. 
 
\end{proof}
\begin{proof}[Proof of Lemma \ref{mainestimate}] 
The function $w_j := u_j - v_j$ is harmonic in $\square$ and so one has
 \begin{flalign*}
 	\langle f_i , u_j \rangle -  	\langle f_i , v_j \rangle & =  	\langle \chi_{R_\delta}f_i,w_j \rangle_{B_\delta} =  	\langle -\Delta v_i,w_j \rangle_{B_\delta} \\
 	&  =  	\langle \nabla v_i, \nabla w_j  \rangle_{B_\delta}  -   	\langle \partial_r v_i,w_j \rangle_{S_\delta} \\
 	& \stackrel{\eqref{harm2}}{=}  	\langle \nabla v_i, \nabla w_j  \rangle_{B_\delta}  +   	\langle v_i,\partial_r w_j \rangle_{S_\delta} + d_i \overline{w_j(0)} \\
 	& = 2 	\langle \nabla v_i, \nabla w_j  \rangle_{B_\delta}   + d_i \overline{w_j(0)}.
 \end{flalign*}
The last equality is from the divergence theorem and the fact that $w_j$ is harmonic in $B_\delta$. Now,
\begin{flalign}
w_j(0) = u_j(0) - v_j(0)  & =   	\langle \chi_{R_\delta}f_j ,  G_\t \rangle_\square - 	\langle  \chi_{R_\delta}f_j ,  - \tfrac{1}{2 \pi} \ln{r} \rangle_{\mathbb{R}^2}
=   	\langle \chi_{R_\delta}f_j ,  g_\t \rangle_{B_\delta}  = 	\langle - \Delta v_j ,  g_\t \rangle_{B_\delta}  \nonumber  \\
& = 	\langle \nabla v_j,  \nabla g_\t \rangle_{B_\delta} - 	\langle \partial_r v_j ,  g_\t \rangle_{S_\delta}  \nonumber  \\
&\stackrel{\eqref{harm2}}{=}	\langle \nabla v_j,  \nabla g_\t \rangle_{B_\delta} + 	\langle v_j , \partial_r g_\t \rangle_{S_\delta} + d_j \overline{g_\t(0)}   \nonumber \\
& = 2 \langle \nabla v_j,  \nabla g_\t \rangle_{B_\delta} + d_j \overline{g_\t(0)}. \label{15.4.24e3}
\end{flalign}
 Therefore
\begin{flalign*}
| \langle f_i , u_j \rangle - \langle f_i, v_j \rangle - g_\t(0) d_i\overline{d_j}| & = |2 	\langle \nabla v_i, \nabla w_j  \rangle_{B_\delta}  + 2d_i  \langle \nabla g_\t,  \nabla v_j \rangle_{B_\delta}| \\
& \le 2 	\| \nabla v_i\|_{B_\delta}\| \nabla w_j  \|_{B_\delta}  + 2|d_i|  \| \nabla g_\t\|_{B_\delta}  \|\nabla v_j \|_{B_\delta} \\
&\le 2 	\| \nabla v_i\|_{B_\alpha}\| \nabla w_j  \|_{B_\delta}  + 2|d_i|  \| \nabla g_\t\|_{B_\delta}  \|\nabla v_j \|_{B_\alpha} \\
&\stackrel{\eqref{harm1}}{\le} 2 \frac{\delta}{\alpha}\Big(\| \nabla v_i\|_{B_\alpha}\| \nabla w_j  \|_{B_\alpha}  + |d_i|  \| \nabla g_\t\|_{B_\alpha}  \|\nabla v_j \|_{B_\alpha} \Big)\\
&\le 2 \frac{\delta}{\alpha} \sqrt{ \|\nabla v_i \|_{B_\alpha}^2+ |d_i|^2\|\nabla g_\t \|_{B_\alpha}^2}\sqrt{ \|\nabla v_j \|_{B_\alpha}^2+ \|\nabla w_j \|_{B_\alpha}^2}.
\end{flalign*}
Now, by \eqref{19.3.24e1} and \eqref{19.3.24e2} one has 
\begin{equation}\label{15.4.24e4}
\| \nabla v_i \|_{L^2(B_\alpha)} ^2 + |d_i|^2 \| \nabla g_\t \|_{L^2(B_\alpha)} ^2 \le\langle f_i,  v_i \rangle+ |d_i|^2  g_\t(0).	
\end{equation}
Therefore, 
to prove the lemma, it suffices to show 
\begin{equation}
	\| \nabla v_j \|^2_{B_\alpha} + \| \nabla w_j \|^2_{B_\alpha} \le 2 \Big( \langle f_j,v_j\rangle + |d_j|^2 g_\t(0) \Big).\label{27.3.24e1}
\end{equation}

To this end, note  that
 \begin{flalign}
0 &= \langle -\Delta u_j, u_j\rangle_{\square \backslash B_\alpha} = \| \nabla u_j \|^2_{\square \backslash B_\alpha} + \langle \partial_r u_j, u_j\rangle_{S_\alpha} \nonumber \\
& = \| \nabla u_j \|^2_{\square \backslash B_\alpha} + \langle \partial_r v_j, v_j\rangle_{S_\alpha} + \langle \partial_r w_j, w_j\rangle_{S_\alpha} + \langle \partial_r v_j, w_j\rangle_{S_\alpha} + \langle \partial_r w_j,v_j \rangle_{S_\alpha}      \nonumber\\
& = \| \nabla u_j \|^2_{\square \backslash B_\alpha} +  \| \nabla v_j \|^2_{B_\alpha} - \langle f_j,  v_j \rangle + \| \nabla w_j \|^2_{ B_\alpha} +  \langle \partial_r v_j, w_j\rangle_{S_\alpha} + \langle \partial_r w_j,v_j \rangle_{S_\alpha} .  \label{5.3.24e1}
 \end{flalign}
 Now, by \eqref{harm2}
 \[
 \langle \partial_r v_j, w_j\rangle_{S_\alpha} + \langle \partial_r w_j,v_j\rangle_{S_\alpha}  = 2 \mathrm{i}\, \mathrm{Im} \langle  \partial_r w_j, v_j\rangle_{S_\alpha} - d_j \overline{w_j(0)}.  
 \]
 Therefore,  taking the real part of \eqref{5.3.24e1}, and noting that $ \langle f_j,  v_j \rangle$ is real (cf. \eqref{19.3.24e1}), gives
 \[
 0 = \| \nabla u_j \|^2_{\square \backslash B_\alpha} +  \| \nabla v_j \|^2_{B_\alpha} - \langle f_j,  v_j \rangle + \| \nabla w_j \|^2_{ B_\alpha} - \mathrm{Re} (d_j \overline{w_j(0)} );
 \]
 i.e.
 \begin{flalign*}
\| \nabla v_j \|^2_{B_\alpha}  + \| \nabla w_j \|^2_{ B_\alpha} & \le  \langle f_j,  v_j \rangle  + \mathrm{Re} (d_j \overline{w_j(0)} )\nonumber \\
& \stackrel{\eqref{15.4.24e3}}{=} \langle f_j,  v_j \rangle  +  |d_j|^2 g_\t(0) + \mathrm{Re} (2d_j  \langle \nabla g_\t, \nabla v_j \rangle_{B_\delta} ) \nonumber \\
&  \le  \langle f_j,  v_j \rangle  +  |d_j|^2 g_\t(0) + |d_j|^2  \| \nabla g_\t\|_{B_\delta}^2 + \| \nabla v_j \|_{B_\delta}^2  \nonumber \\
& \stackrel{\eqref{15.4.24e4}}{\le} 2 \Big( \langle f_j,  v_j \rangle + |d_j|^2 g_\t(0) \Big). 
 \end{flalign*}
That is, \eqref{27.3.24e1} holds and the proof is complete.
\end{proof}
Lemma \ref{mainestimate} implies the following inequalities:
\begin{lem}\label{lem.denom1}
For any $ \alpha \in(\delta,\pi)$, $i,j,k,l \in \{ 1, 2\}$, one has
\begin{align}\label{15.4.24e6}
&\big| \langle u_{i,j} , u_{k,l} \rangle - \langle v_{i,j} , v_{k,l} \rangle \big| \le 2 \sqrt{3}  \frac{\delta}{\alpha} \sqrt{\langle f_i, v_i \rangle + g_\t(0) |d_i|^2  }\sqrt{\langle f_k, v_k \rangle + g_\t(0) |d_k|^2  }; \\ 
\label{15.4.24e7}
&	| \langle u_{k,ij} , \tilde{f}_l \rangle - \langle v_{k,ij} , \tilde{f}_l \rangle  | \le 3 \frac{\delta}{\alpha} \| f_k\|_{L^2(R_\delta)} \| \tilde{f}_l \|_{L^2(R_\delta)}.
\end{align}
\end{lem}
\begin{proof}

\text{Proof of \eqref{15.4.24e6}.} As $u_i = v_i + w_i$, then
\[
\langle u_{i,j} , u_{k,l} \rangle    =  \langle v_{i,j} , v_{k,l} \rangle +   \langle v_{i,j} , w_{k,l} \rangle +  \langle w_{i,j} , u_{k,l} \rangle,
\]
so
\begin{flalign*}
| \langle u_{i,j} , u_{k,l} \rangle - \langle v_{i,j} , v_{k,l} \rangle | & \le   | \langle v_{i,j} , w_{k,l} \rangle| + | \langle w_{i,j} , u_{k,l} \rangle| \\
& \le \| v_{i,j}\|_{B_\alpha} \| w_{k,l}\|_{B_\delta} + \| w_{i,j}\|_{B_\delta} \| u_{k,l}\|_{B_\alpha} \\
& \stackrel{\eqref{harm1}}{\le} \frac{\delta}{\alpha}\big(\| v_{i,j}\|_{B_\alpha} \| w_{k,l}\|_{B_\alpha} + \| w_{i,j}\|_{B_\alpha} \| u_{k,l}\|_{B_\alpha} \big)\\
& \le \frac{\delta}{\alpha} \sqrt{ \| v_{i,j}\|_{B_\alpha}^2 +  \| w_{i,j}\|_{B_\alpha}^2} \sqrt{ \| w_{k,j}\|_{B_\alpha}^2+ \| u_{k,j}\|_{B_\alpha}^2} \\
& \le \sqrt{3} \frac{\delta}{\alpha}  \sqrt{ \| v_{i,j}\|_{B_\alpha}^2 +  \| w_{i,j}\|_{B_\alpha}^2} \sqrt{ \| v_{k,j}\|_{B_\alpha}^2+ \| w_{k,j}\|_{B_\alpha}^2},
\end{flalign*}
since $\| u_{k,l} \|^2 \le 2(\|v_{k,l}\|^2 + \|w_{k,l}\|^2)$.
Now, \eqref{27.3.24e1} states
\[
\| \nabla v_{i} \|_{B_\alpha}^2 +\| \nabla w_{i} \|_{B_\alpha}^2 \le 2 \Big( \langle f_i, v_i \rangle + g_\t(0) |d_i|^2  \Big) .
\]
Therefore,
\[
| \langle u_{i,j} , u_{k,l} \rangle - \langle v_{i,j} , v_{k,l} \rangle |  \le 2 \sqrt{3}  \frac{\delta}{\alpha}\sqrt{\langle f_i, v_i \rangle + g_\t(0) |d_i|^2  }\sqrt{\langle f_k, v_k \rangle + g_\t(0) |d_k|^2  }
\].

\textit{Proof of \eqref{15.4.24e7}.} It is enough to prove the result for $f_k, \tilde{f}_l \in C^\infty_0(R_\delta)$. Here,   $u_{k,i} = A_\t (f_{k,i})$  and $v_{k,i} = A_\delta(f_{k,i})$ for $i=1,2$. This observation, and integration by parts, gives
\[
 \langle u_{k,ij},\tilde{f}_l \rangle  =   \langle A_\t (f_{k,i}), -\tilde{f}_{l,,j} \rangle,
\]
and 
\[
\langle v_{k,ij},\tilde{f}_l \rangle  =\langle A_\delta(f_{k,i}), -\tilde{f}_{l,,j} \rangle.
\]
 Now, we shall utilise Lemma \ref{mainestimate} for $f_1 = f_{k,i}$ and $f_2 = -\tilde{f}_{l,,j}$. Notice that, for such $f_i$, one has $d_i=0$.

 Therefore, Lemma \ref{mainestimate} and the above considerations lead to
 \begin{equation*}
| \langle u_{k,ij},\tilde{f}_l \rangle - \langle v_{k,ij}, \tilde{f}_l \rangle  | \le 3 \frac{\delta}{\alpha} \sqrt{\langle f_{k,i}, (A_\delta f_{k})_{,i}\rangle} \sqrt{\langle \tilde{f}_{l,,j}, (A_\delta \tilde{f}_l)_{,j}\rangle}.
 \end{equation*}
 It remains to note
 \[
 \langle g_{,i}, (A_\delta g)_{,i}\rangle \le \langle g, g\rangle, \quad g \in C^\infty_0(R_\delta), \ i=1,2.
 \]
Indeed, since $ \langle g_{,i}, (A_\delta g)_{,i}\rangle \ge 0$, then 
\[
\langle g_{,i}, (A_\delta g)_{,i}\rangle \le  \langle g_{,1}, (A_\delta g)_{,1}\rangle +  \langle g_{,2}, (A_\delta g)_{,2}\rangle   = -\langle g,  \Delta (A_\delta g)\rangle =  \langle g,  g\rangle.
\]
The desired result now follows.
\end{proof}

We are now in a position to prove   Lemma \ref{lem.prob1}.
\begin{lem}\label{lem.numerator2}
For any $f \in L^2(R_\delta;\CC^2) $, one has
\begin{align}
\label{numerator2} 
|\mathfrak{b}_\t[f] - b^{(\delta)}[f]| & \le \gamma \delta \,  b^{(\delta)}[f] ; \\ 
\label{denom2}| c_\t[f] - c_{\delta,\t}[f]	 | &\le\left(1 + 3 \gamma \right) \delta c_{\delta,\t}[f].
\end{align}

\end{lem}
\begin{proof} \textit{Proof of \eqref{numerator2}.}
By \eqref{20.6.24e1} (for $R=R_\delta$), \eqref{20.6.24e3} and \eqref{15.4.24e7}, one has 
\[
| b_\t[f]  - b^{(\delta)}[f] |  \le 3\frac{\delta}{\alpha} \big( \| f_1 \|^2 + 2 \| f_1 \| \| f_2 \| + \|f_2 \|^2 \big) = 3\frac{\delta}{\alpha} \| f\|^2 \le 3\frac{\delta}{\alpha}  \gamma b^{(\delta)}[f],
\]
for any $\alpha\in (\delta,\pi)$. 

\textit{Proof of \eqref{denom2}.} By Lemma \ref{mainestimate} and \eqref{cdform}, one has 
\begin{equation*}
	\begin{aligned}
\Big| \sum_{i=1}^2\left( \langle f_i, A_\t f_i \rangle  - \big(\langle f_i, A_\delta f_i \rangle + g_\t(0)|d_i|^2 \big) \right)\Big|& \le 3 \frac{\delta}{\alpha} \sum_{i=1}^2  \Big( \langle f_i, A_\delta f_i \rangle  + g_\t(0) |d_i|^2\Big) \\& \le 3 \frac{\delta}{\alpha} c_{\delta,\theta}[f].
	\end{aligned}
\end{equation*}
Inequality \eqref{15.4.24e6} gives
\[
\begin{aligned}
	\big| \langle (A_\t f_i)_{,j} , (A_\t f_i)_{,j} \rangle -  \langle (A_\delta f_i)_{,j} , (A_\delta f_i)_{,j} \rangle\big| &\le 2 \sqrt{3}  \frac{\delta}{\alpha}  \big(\langle f_i, A_\delta f_i \rangle + g_\t(0) |d_i|^2\big),
\end{aligned}
\]
and
\[
\begin{aligned}
\big| \langle (A_\t f_1)_{,i} , (A_\t f_2)_{,j} \rangle -  \langle (A_\delta f_1)_{,i} , (A_\delta f_2)_{,j} \rangle\big| &\le 2 \sqrt{3} \frac{\delta}{\alpha}  \sqrt{\langle f_1, A_\delta f_1 \rangle + g_\t(0) |d_1|^2  }\sqrt{\langle f_2, A_\delta f_2 \rangle + g_\t(0) |d_2|^2  }\\
& \le \sqrt{3}  \ \frac{\delta}{\alpha}   c_{\delta,\theta}[f].
\end{aligned}
\]
So, the above assertions, \eqref{20.6.24e2} and \eqref{cdform}, give
\[
| c_\t[f] - c_{\delta,\theta}[f] | \le \left(3 + 4\sqrt{3} \gamma \right) \left(\frac{\delta}{\alpha}\right) c_{\delta,\theta}[f],
\]
for $\alpha \in (\delta, \pi)$. 
\end{proof}

\begin{proof}[Proof of Lemma \ref{lem.prob1}] 
\begin{flalign*}
	\left| \frac{b_\t[f]}{c_\t[f]} - \frac{b^{(\delta)}[f]}{c_{\delta,\theta}[f]} \right| & \le \left| \frac{b_\t[f]}{c_\t[f]} - \frac{b_\t[f]}{c_{\delta,\theta}[f]} \right| +  \left| \frac{b_\t[f]}{c_{\delta,\theta}[f]} - \frac{b^{(\delta)}[f]}{c_{\delta,\theta}[f]} \right|  \\
	 & =\left( \frac{b_\t[f]}{c_\t[f]c_{\delta,\theta}[f]} \right) \big| c_\t[f] - c_{\delta,\theta}[f] \big| + \left| \frac{b_\t[f] - b^{(\delta)}[f]}{c_{\delta,\theta}[f]} \right|   \\
	& \le\left(1 + 3 \gamma \right) \delta\frac{b_\t[f]}{c_\t[f]}  +\gamma \delta \frac{b^{(\delta)}[f]}{c_{\delta,\theta}[f]}
\end{flalign*}
\end{proof}
\section*{Acknowledgements}
The authors would like to thank Professor Valery Smyshlyaev for interesting discussions and useful suggestions related to this problem. 

\end{document}